\documentclass[letterpaper, 10 pt, conference]{ieeeconf}
\IEEEoverridecommandlockouts                              
\usepackage{cite}
\usepackage{amsmath,amssymb,amsfonts}
\usepackage{graphicx}
\usepackage{float}

\makeatletter
\let\NAT@parse\undefined
\makeatother
\usepackage[hidelinks]{hyperref}
\usepackage{textcomp}

\usepackage{dq_macros}
\usepackage{booktabs}
\usepackage{array}      

\usepackage[dvipsnames]{xcolor}
\usepackage[normalem]{ulem}
\usepackage[colorinlistoftodos]{todonotes}  

\newcommand{\kamedit}[1]{%
\textcolor{red}{#1}}

\usepackage{algorithm, algpseudocodex}

\algrenewcommand\algorithmicprocedure{\textbf{method}}
\algdef{SE}[CLASS]{Class}{EndClass}[1]%
  {\textbf{class} \textsc{#1}}%
  {}

\usepackage{circuitikz} 
\usepackage[capitalise]{cleveref} 

\begin{document}
\title{\LARGE \bf Observability Conditions and Filter Design for Visual Pose Estimation via Dual Quaternions$^*$}
\author{Nicholas B. Andrews$^{1}$ and Kristi A. Morgansen$^{1}$%
\thanks{$^*$This work was supported in part by Blue Origin Enterprises, L.P.}
\thanks{$^{1}$Department of Aeronautics and Astronautics, University of Washington, Seattle, WA 98195 USA  
{\tt\small [{\href{mailto:nian6018@uw.edu}{nian6018},  \href{mailto:morgansn@uw.edu}{morgansn}}]@uw.edu}}}

\maketitle
\thispagestyle{empty}
\pagestyle{empty}

\maketitle

\begin{abstract}
This paper presents a dual quaternion framework for 6-DOF visual target tracking that addresses key limitations of perspective-n-point (\pnp) solvers: sensitivity to noise and outliers, and inability to propagate estimates through measurement dropouts. A nonlinear observability analysis is performed using a Lie algebraic approach, deriving sufficient conditions for local observability under two sensing modalities: relative position vector and unit vector measurements. For the unit vector case, the classical collinear feature point degeneracy of the perspective-three-point problem is recovered through rank analysis of the observability codistribution matrix, providing a control-theoretic interpretation of a previously geometric result. A dual quaternion Lie group unscented Kalman filter is then developed, directly modeling relative dynamics without assumptions about cooperative measurements or slowly-varying motion. Simulations demonstrate improved pose estimation accuracy and robustness to occlusions compared to an off-the-shelf \pnp solver. Results are broadly applicable to visual-inertial navigation, simultaneous localization and mapping, and \pnp solver development.
\end{abstract}

\section{Introduction}
\label{sec:introduction}

{Accurate} state estimation is essential for the safe and efficient operation of autonomous robotic systems, including spacecraft, underwater vehicles, and aircraft. Tasks such as manipulating, docking, and avoiding obstacles critically depend on precise relative pose (position and orientation) information for guidance, navigation, and control. Although autonomous systems typically utilize multiple sensing modalities, the availability and decreasing cost of high-performance cameras, combined with advances in computer vision and onboard processing, have driven increased adoption of vision-based sensing in closed-loop control systems.

Central to vision-based pose estimation is the Perspective-$n$-Point (\pnp) problem: estimating the pose of a camera relative to $n$ known 3-D feature points whose projections are observed in the 2-D image plane. While typically formulated in the context of computer vision and cameras, the \pnp framework is mathematically equivalent to bearing only or unit vector position measurements that arise in applications using acoustic or radio frequency sensors for localization. Numerous algorithms have been developed to find solutions to the \pnp problem \cite{Lepetit2009-so, Ke2017-rn, terzakis2020}, offering practical ``off-the-shelf'' utility due to their simplicity of implementation and widespread adoption in computer vision libraries. However, classical \pnp solvers exhibit two key limitations for target tracking applications. First, they are deterministic and inherently sensitive to measurement noise and outliers. Second, they are memoryless: they do not incorporate knowledge of rigid body motion, rely heavily on consistent feature matching and outlier rejection, and cannot propagate state estimates through measurement dropouts or occlusions.

To address these shortcomings, practitioners often augment \pnp solvers with sequential stochastic state estimators, such as Kalman and particle filters, which leverage physics-informed motion models. These methods have demonstrated reduced estimation error and greater robustness to noise and outliers, even with fewer correlated feature points \cite{Sveier2021-ev, Mehralian2020-ns}. While reducing the number of required features alleviates some computational burden, the filtering algorithms themselves introduce additional overhead that may be prohibitive for platforms with limited onboard compute.

The limitations of traditional \pnp solvers and advantageous properties of sequential stochastic state estimators motivate continued research in the development of pose estimation frameworks that exploit both the geometric constraints of the \pnp problem and the rigid body physics. Dual quaternions provide a particularly attractive mathematical foundation for such frameworks, offering a compact, singularity-free representation of \se states that unifies rotation and translation within a single algebraic structure. Dual quaternions have been shown to be the most efficient representation to perform basic pose transformations in terms of storage requirements and number of operations~\cite{Funda1990}. Moreover, many dual quaternion operations admit a matrix representation, enabling the application of linear algebraic tools and simplifying the analysis of 6-DOF systems \cite{Andrews2024-jk}. Due to their advantageous computation and modeling properties, dual quaternions have been used for controller \cite{Reynolds2018-ou, Filipe2015-zc} and estimator \cite{Hudoba-de-Badyn2025-bk, Razgus2017-zq, Filipe2015-ep, Zivan2022-pw, Sveier2021-ev} design. 

The primary motivation for the work here is a comprehensive treatment of the modeling, nonlinear observability analysis, and sequential state estimator design for 6-DOF visual target tracking approached entirely through the use of dual quaternions. In this work{, no assumptions are made} about additional measurement sources (inertial measurement units, gyroscopes, etc.) or the cooperative nature between the target and the camera. 
By making as few assumptions as possible about additional measurement types, relative motion dynamics, and cooperation, the results from this work{ are} broadly applicable and generalizable to applications such as visual-inertial navigation, simultaneous localization and mapping (SLAM), and \pnp solver algorithms.

The analytical observability of the \pnp problem has been investigated in various contexts using several control-theoretic methods. In \cite{Hamel2018-ib, Hua2020-xl}, the observability Gramian was calculated to derive observability conditions for a linear time-varying 6-DOF system using bearing and velocity measurements. A similar Gramian-based approach was applied in \cite{Zivan2022-pw} to a satellite relative navigation system with line-of-sight and IMU measurements modeled using dual quaternions. Of the prior works, the present analysis is most closely related to \cite{Sun2002-pb} and \cite{Mirzaei2008-kl}. In \cite{Sun2002-pb}, the authors investigated uniqueness of the \pnp problem across the one-, two-~, \mbox{three-,} and $n$-unit vector measurement cases, deriving conditions on the rank of the Fisher information matrix. Their analysis, however, does not incorporate a process model, relies on a small-angle linearization about the true attitude, and characterizes only instantaneous observability, a stricter requirement than what sequential state estimators require, since such estimators propagate prior state information. A Lie algebraic observability analysis in a related setting was performed in \cite{Mirzaei2008-kl} for 6-DOF IMU-camera calibration. However, that work incorporates IMU measurements and does not treat the \pnp geometry directly.

To the best of our knowledge, this work presents the first nonlinear observability analysis of the \pnp problem using a Lie algebraic approach, establishing local weak observability through rank analysis of the observability codistribution matrix. The analysis is carried out for two sensing modalities: relative position vector measurements and unit vector measurements representative of projected camera observations. Notably, the observability conditions for the unit vector case recover the classical collinear-feature degeneracy known from the Perspective-Three-Point (\ptp) literature, a result typically derived through algorithmic and geometric arguments, here obtained from a control-theoretic perspective. In contrast to \cite{Sun2002-pb}, the present approach incorporates the system motion model, makes no small-angle approximation, and respects the underlying \se geometry. The resulting observability conditions are directly applicable to sequential state estimators and establish that three unit vector measurements suffice to uniquely determine the pose given a prior estimate, providing a theoretical guarantee of practical relevance for filter-based pose estimation.


The second key contribution is a dual quaternion formulation of a Lie group unscented Kalman filter (UKF) \cite{Brossard2017-nf, Zambo2024-gx} which is specifically tailored to the \pnp problem and is used to demonstrate the nonlinear local observability conditions from the first primary contribution. Previous works have developed dual quaternion extended Kalman filters (EKFs) \cite{Hudoba-de-Badyn2025-bk, Razgus2017-zq, Filipe2015-ep, Zivan2022-pw} and particle filters \cite{Sveier2021-ev} for {the} relative pose and velocity estimation problem. However, these works either assume the relative velocity is measured directly via cooperative gyroscopes or lasers, is constant, or is slowly varying and estimated as a random walk. In{ the} formulation {here},{ none} of these assumptions {are made,} and a model of the relative dynamics {is incorporated} directly, allowing for improved pose and velocity estimation through a complete modeling of rigid body motion. Additionally, while the UKF is computationally more demanding than an EKF, it is known to perform better for highly nonlinear systems, does not require the calculation of Jacobian matrices which may be analytically intractable for some systems, and is much less computationally demanding than a particle filter which is oftentimes infeasible to run on mobile systems without a GPU.

The remainder of this article is organized as follows. Background material on quaternions and dual quaternions for modeling states in \so and \se is reviewed in Section~\ref{sec:quaternion}. The dual quaternion relative motion dynamics are presented at the end of Section~\ref{sec:model} in the context of a camera observing a target. A brief introduction to nonlinear observability theory and the matrix rank condition for local observability is provided in Section~\ref{sec:nlobsv}. Sufficient conditions for the relative pose and velocity states to be locally observable with only relative position measurements and with only unit vector measurements are derived in Sections~\ref{sec:position_analysis} and \ref{sec:proj_analysis}, respectively. A dual quaternion formulation of a Lie algebraic UKF, with corresponding pseudocode and code repository, is presented in Section~\ref{sec:estimator}, with additional details provided in the Appendix. The previously derived observability conditions are demonstrated in Section~\ref{sec:simulation} using the dual quaternion unscented Kalman filter in a simple simulation with dynamic camera and target motion, highlighting improved pose estimation accuracy and robustness during occlusions relative to a traditional \pnp solver implemented in OpenCV. Conclusions and future research directions are discussed in Section~\ref{sec:conclusion}.

\section{Quaternion Framework} \label{sec:quaternion}
In this section, the fundamentals of the quaternion framework which will serve as a basis for the relative motion model, observability analyses, and filter design in subsequent sections are introduced. The quaternion and dual quaternion definitions and properties presented here are summarized from \cite{Reynolds2018-ou, Valverde2018-ri, Filipe2015-zc, Sola2017-mp}.
\begin{table*}
    \centering
    \renewcommand{\arraystretch}{1.7}
    \caption{Quaternion and dual quaternion operations and element definitions \cite{Filipe2015-zc, Stanfield2021-bt}.}
    \label{tab:quaternion}
    \begin{tabular}{
        l
        >{\centering\arraybackslash}m{0.35\textwidth}
        >{\centering\arraybackslash}m{0.35\textwidth}
    }
        \toprule
        \textbf{Operation}
            & \textbf{Quaternion Definition}
            & \textbf{Dual Quaternion Definition} \\
        \midrule
        Addition
            & $\qblank{a} + \qblank{b} = \mtx{\qscalar{a} + \qscalar{b} & \qvec{a} + \qvec{b}}$
            & $\dualqblank{a} + \dualqblank{b} = \qreal{a} + \qreal{b}+ \dualunit (\qdual{a} + \qdual{b})$ \\
        Scalar Multiplication
            & $\lambda \qblank{a} = \mtx{\lambda \qscalar{a} & \lambda \qvec{a}}$
            & $\lambda \dualqblank{a} = \lambda \qreal{a} + \dualunit \lambda \qdual{a}$ \\
        Multiplication
            & $\qblank{a} \qblank{b} = \mtx{ \qscalar{a} \qscalar{b} - \qvec{a} \qdot \qvec{b} & \qscalar{a} \qvec{b} + \qscalar{b} \qvec{a} + \qvec{a} \cross \qvec{b}}$
            & $\dualqblank{a} \dualqblank{b} = \qreal{a} \qreal{b} + \dualunit (\qdual{a} \qreal{b} + \qreal{a} \qdual{b})$ \\
        Conjugate
            & $\qconj{\qblank{a}} = \mtx{\qscalar{a} & -\qvec{a}}$
            & $\qconj{\dualqblank{a}} = \qconj{\qreal{a}} + \dualunit \qconj{\qdual{a}}$ \\
        Dot Product
            & $\qblank{a} \qdot \qblank{b} = \mtx{\qscalar{a} \qscalar{b} + \qvec{a} \qdot \qvec{b} & \zeros{3}{}}$
            & $\dualqblank{a} \qdot \dualqblank{b} = \qreal{a} \qdot \qreal{b} + \dualunit (\qdual{a} \qdot \qreal{b} + \qreal{a} \qdot \qdual{b})$ \\
        Cross Product
            & $\qblank{a} \cross \qblank{b} = \mtx{0 & \qscalar{a} \qvec{b} + \qscalar{b} \qvec{a} + \qvec{a} \cross \qvec{b}}$
            & $\dualqblank{a} \cross \dualqblank{b} = \qreal{a} \cross \qreal{b} + \dualunit (\qdual{a} \cross \qreal{b} + \qreal{a} \cross \qdual{b})$ \\
        Norm
            & $\norm{\qblank{a}} = \sqrt{\qscalar{a}^2 + \qvec{a} \qdot \qvec{a}}$
            & $\norm{\dualqblank{a}} = \sqrt{\norm{\qreal{a}}^2 + \norm{\qdual{a}}^2}$ \\
        Real Part   & Undefined  & $\qreal{(\dualqblank{a})} = \qreal{\qblank{a}}$ \\
        Dual Part   & Undefined  & $\qdual{(\dualqblank{a})} = \qdual{\qblank{a}}$ \\
        Zero Element   & $\qzero = \mtx{0 & \zeros{3}{}}$  & $\dualzero = \qzero + \dualunit \qzero$ \\
        Identity Element   & $\qone = \mtx{1 & \zeros{3}{}}$  & $\dualone = \qone + \dualunit \qzero$ \\
        \bottomrule
    \end{tabular}
\end{table*}

\subsection*{Quaternions}
A quaternion object is defined by its coefficients $\qscalar{q}, \qblank{q_1}, \qblank{q_2}, \qblank{q_3} \in \reals{}$ and basis elements $\basis{i}, \basis{j}, \basis{k}$:
\begin{align}
    \qblank{q} = \qscalar{q} + \qblank{q_1} \basis{i} + \qblank{q_2} \basis{j} + \qblank{q_3} \basis{k}.
\end{align}
The basis elements satisfy the properties \mbox{$\basis{i}^2 = \basis{j}^2 = \basis{k}^2 = -1$}, \mbox{$\basis{i} = \basis{j}\basis{k} = -\basis{k}\basis{j}$}, \mbox{$\basis{j} = \basis{k}\basis{i} = -\basis{i}\basis{k}$}, \mbox{$\basis{k} = \basis{i}\basis{j} = -\basis{j}\basis{i}$}, and \mbox{$\basis{i}\basis{j}\basis{k} = -1$}. The sets of quaternions, scalar quaternions, and vector quaternions are defined as \mbox{$\quats = \left\{\qblank{q} \st  \qscalar{q} + \qblank{q_1} \basis{i} + \qblank{q_2} \basis{j} + \qblank{q_3} \basis{k} \right\}$}, \mbox{$\quatss = \left\{\qblank{q} \in \quats \st \qblank{q_1}, \qblank{q_2}, \qblank{q_3} = 0 \right\}$}, and \mbox{$\quatsv = \left\{\qblank{q} \in \quats \st \qscalar{q} = 0 \right\}$}. The elements of a quaternion are typically divided into a scalar component, $\qscalar{q}$, and vector component, \mbox{$\qvec{q} =  \mtx{ \qblank{q_1} & \qblank{q_2} & \qblank{q_3} } \in \reals{3}$}, allowing a quaternion to be succinctly expressed as the concatenated array \mbox{$\qblank{q} = \mtx{\qscalar{q} & \qvec{q} } \in \quats$}. In some definitions, the scalar component is the last entry in the quaternion array, however, for this work the scalar component will always be the first entry. Quaternion operations and element definitions are defined in Table~\ref{tab:quaternion}.

Unit quaternions form a Lie group that is isomorphic to \su and are an advantageous \so representation compared to Euler angles because they are singularity-free and do not suffer from ``gimbal lock'' degenerate orientations. The set of unit quaternions is defined as \mbox{$\quatsu = \left\{ \qblank{q} \in \quats \st  \qconj{\qblank{q}} \qblank{q} = \qblank{q} \qconj{\qblank{q}} = \qblank{q} \qdot \qblank{q} = \qone \right\}$}. A unit quaternion $\qblank{q}_u \in \quatsu$ can be expressed as a rotation angle, $\phi \in \reals{}$, about a unit vector, $\qvec{n} \in \reals{3}$:
\begin{align}
    \qblank{q}_u = \mtx{ \cos \left( \frac{\phi}{2} \right) & \sin \left( \frac{\phi}{2} \right) \qvec{n} }. 
\end{align}
The orientation of a frame $\cf{x}$ with respect to a frame $\cf{y}$ is represented by the {unit quaternion} \mbox{$\q{x}{y} \in \quatsu$}. Unit quaternions have the conjugate inverse property
\begin{gather}
    \label{eq:unitq}
    \inv{\q{x}{y}} = \qconj{\q{x}{y}} \triangleq \q{y}{x},
\end{gather}
admitting a convenient form for transforming a vector in $\reals{3}$ between coordinate frames. By lifting a vector $\qvecf{v}{x} \in \reals{3}$ represented in frame $\cf{x}$ coordinates to a vector quaternion \mbox{$\qblankf{v}{x} = \mtx{ 0 & \qvecf{v}{x} } \in \quatsv$}, coordinate transformations to and from frame $\cf{y}$ have the form
\begin{salign}
    \qblankf{v}{y} &= \qconj{\q{y}{x}} \qblankf{v}{x} \q{y}{x}, \\
    \qblankf{v}{x} &= \q{y}{x} \qblankf{v}{y} \qconj{\q{y}{x}}. 
\end{salign}
Additionally, unit quaternions can be chained together to solve for the total relative rotation between multiple reference frames: 
\begin{gather}
    \q{x}{y} = \qconj{\q{y}{z}} \q{x}{z}. 
\end{gather}

The unit quaternion kinematic equation is
\begin{align}
    \label{eq:qkin}
    \dq{x}{y} = \frac{1}{2} \q{x}{y} \qomega{x}{y}{x} = \frac{1}{2} \qomega{x}{y}{y} \q{x}{y}, 
\end{align}
where $\omeg{x}{y}{x} \in \reals{3}$ is the angular velocity of frame~$\cf{x}$ relative to frame~$\cf{y}$ expressed in frame $\cf{x}$ coordinates, and \mbox{$\qomega{x}{y}{x} = \mtx{0 & \omeg{x}{y}{x}} \in \quatsv$}.

\subsubsection*{Matrix Form}
Many quaternion operations can be rewritten in a matrix form which allows for easier manipulation and use of matrix calculus tools for quaternion calculus. A quaternion $\qblank{q} \in \quats$ left multiplied by a matrix \mbox{$\mblank{M} \in \reals{4 \times 4}$} follows the standard matrix multiplication algebra as if $\qblank{q}$ were a $\reals{4 \times 1}$ vector:
\begin{sgather}
    \mblank{M} = \mtx{m & \mblank{m}_1 \\ \mblank{m}_2 & \mblank{M}_3} \in \reals{4 \times 4},  \\
    \mblank{M} \qblank{q} = \mtx{m \qscalar{q} + \mblank{m}_1 \qvec{q}, \ \mblank{m}_2 \qscalar{q} + \mblank{M}_3 \qvec{q}} \in \quats,
\end{sgather}
where $m \in \reals{}, \ \mblank{m}_1 \in \reals{1 \times 3}, \ \mblank{m}_2 \in \reals{3 \times 1}$, and $\ \mblank{M}_3 \in \reals{3 \times 3}$.

The left, $\qleft{q}$, and right, $\qright{q}$, quaternion multiplication matrices are defined as:
\begin{salign}
    \skewmat{q} &= \mtx{0 & -\qblank{q}_3 & \qblank{q}_2 \\
    \qblank{q}_3 & 0 & -\qblank{q}_1 \\
    -\qblank{q}_2 & \qblank{q}_1 & 0} \in \reals{3 \times 3},  \\
    \qleft{q} &= \qscalar{q} \eye{4} + \mtx{0 & -\tpose{\qvec{q}}  \\
    \qvec{q} & \skewmat{q}} \in \reals{4 \times 4},  \\
    \qright{q} &= \qscalar{q} \eye{4} + \mtx{0 & -\tpose{\qvec{q}} \\
    \qvec{q} & -\skewmat{q}} \in \reals{4 \times 4}, 
\end{salign}
where $\eye{n}$ is the $n \times n$ identity matrix and $\qvec{q}$ is interpreted as a vector {in $\reals{3 \times 1}$} for consistent block matrix composition. Quaternion multiplication can then be re-written into a matrix form as
\begin{sgather}
    \qblank{a} \qblank{b} = \qleft{a} \qblank{b} = \qright{b} \qblank{a},  \\
    \qblank{a} \qblank{b} \qblank{c} = \qleft{\qblank{a} \qblank{b}} \qblank{c} = \qright{\qblank{b} \qblank{c}} \qblank{a} 
\end{sgather}
for $\qblank{a}, \qblank{b}, \qblank{c} \in \quats$. Note that attention must be exercised to mind the order of operations when transforming quaternions to and from the matrix form, as demonstrated by the following example:
\begin{gather}
     \qblank{a} \qblank{b} \qblank{c} = \left( \qright{\qblank{b}} \qblank{a} \right) \qblank{c} \neq \qright{\qblank{b}} \left( \qblank{a}  \qblank{c} \right). 
\end{gather}
Additionally, the quaternion cross product can be reconstructed as the matrix
\begin{gather}
    \qcross{q} = \mtx{0 & \zeros{1}{3} \\
    \qvec{q} & \qscalar{q} \eye{3} + \skewmat{q}},  
\end{gather}
such that $\qblank{a} \cross \qblank{b} = \qcross{a} \qblank{b}$.

The quaternion conjugate is antihomomorphic
\begin{gather}
    \qconj{\left( \qblank{a} \qblank{b} \right)} = \qconj{\qblank{b}} \qconj{\qblank{a}} 
\end{gather}
and can be expressed in matrix form using the conjugate matrix $\eyeconj = \diag{\mtx{1 & -1 & -1 & -1}} \in \reals{4 \times 4}$:
\begin{gather}
    \qconj{\qblank{q}} = \eyeconj \qblank{q} 
    , 
\end{gather}
where $\diag{\cdot}$ is the diagonal matrix function which places the input vector along the diagonal and off-diagonal elements are zero.


\subsubsection*{Derivatives}
Recasting quaternion expressions in matrix form enables the direct application of matrix calculus tools for differentiation. Some common quaternion derivatives that will be used to construct dual quaternion derivatives in later sections are:
\begin{gather}
    \pd{\qblank{a} \qblank{b}}{\qblank{a}} = \qright{b},  \quad
    \pd{\qblank{a} \qblank{b}}{\qblank{b}} = \qleft{a}, \quad
    \pd{\qconj{\qblank{a}} \qblank{b} \qblank{a}}{\qblank{a}} = \qleft{\qconj{\qblank{a}} \qblank{b}} + \qright{\qblank{b} \qblank{a}} \eyeconj. \label{eq:qderiv}
\end{gather}

\subsection*{Dual Quaternions}
Similar to how a complex number is composed of a real and imaginary component, a dual quaternion is formed by a real part, $\qreal{q} \in \quats$, and a dual part, $\qdual{q} \in \quats$. The dual unit, $\dualunit$, has the properties $\dualunit^2 = 0$ and $\dualunit \neq 0$. The set of dual quaternions, scalar dual quaternions, and vector dual quaternions are defined as $\dquats = \left\{\dualqblank{q} \st \dualqblank{q} = \qreal{q} + \dualunit \qdual{q}, \ \qreal{q}, \qdual{q} \in \quats \right\}$, $\dquatss = \left\{\dualqblank{q} \in \dquats \st \qreal{q}, \qdual{q} \in \quatss \right\}$, and \mbox{$\dquatsv = \left\{\dualqblank{q} \in \dquats \st \qreal{q}, \qdual{q} \in \quatsv \right\}$}. 
Dual quaternion operations and element definitions are given in Table~\ref{tab:quaternion}. 

The dual pose can be thought of as an extension of the unit quaternion and is a dual quaternion which embeds the relative orientation and translation between 6-DOF coordinate frames. The dual pose belongs to the set of {unit dual quaternions} defined as \mbox{$\dquatsu = \left\{\dualqblank{q} \in \dquats \st \qconj{\dualqblank{q}} \dualqblank{q} = \dualqblank{q} \qconj{\dualqblank{q}} = \dualqblank{q} \qdot \dualqblank{q} = \dualone\right\}$}. The pose of a frame $\cf{x}$ with respect to a frame $\cf{y}$ is represented by the {dual pose}, $\dualq{x}{y} \in \dquatsu$, and is defined as
\begin{align}
    \dualq{x}{y} &= \q{x}{y} + \dualunit \frac{1}{2} \qpos{x}{y}{y} \q{x}{y} \label{eq:kinxy} \\
    &= \q{x}{y} + \dualunit \frac{1}{2} \q{x}{y} \qpos{x}{y}{x},
\end{align}
where $\q{x}{y} \in \quatsu$ is the orientation of frame $\cf{x}$ with respect to frame $\cf{y}$, $\pos{x}{y}{x} \in \reals{3}$ is the position of frame $\cf{x}$ relative to frame $\cf{y}$ expressed in frame $\cf{x}$ coordinates, and \mbox{$\qpos{x}{y}{x} = \mtx{ 0 & \pos{x}{y}{x} } \in \quatsv$}.

Many parallels exist between quaternions and dual quaternions in terms of properties and expression forms for changing coordinate frames, kinematics, and derivatives. Similar to unit quaternions, unit dual quaternions also have the conjugate inverse property:
\begin{gather}
    \inv{\dualq{x}{y}} = \qconj{\dualq{x}{y}} \triangleq \dualq{y}{x}.
\end{gather}
Changing coordinate frames between vector dual quaternions $\dualqblankf{v}{x}, \dualqblankf{v}{y} \in \dquatsv$ has the familiar quaternion form
\begin{salign}
    \dualqblankf{v}{y} &= \qconj{\dualq{y}{x}} \dualqblankf{v}{x} \dualq{y}{x}, \\
    \dualqblankf{v}{x} &= \dualq{y}{x} \dualqblankf{v}{y} \qconj{\dualq{y}{x}}. 
\end{salign}
Additionally, unit dual quaternions can be chained together over intermediate pose transformations to solve for the total relative transformation between frames
\begin{gather}
    \dualq{x}{y} = \qconj{\dualq{y}{z}} \dualq{x}{z}. 
\end{gather}

The {dual velocity} is a vector dual quaternion that embeds the relative rotational and translational velocities between coordinate frames and is defined as
\begin{align}
    \dualomega{x}{y}{z} = \qomega{x}{y}{z} + \dualunit (\qvel{x}{y}{z} + \qomega{x}{y}{z} \cross \qpos{z}{x}{z}) \in \dquatsv,  \label{eq:dqvel}
\end{align}
where $\omeg{x}{y}{z} \in \reals{3}$ is the angular velocity of frame $\cf{x}$ relative to frame $\cf{y}$ expressed in frame $\cf{z}$ coordinates, \mbox{$\qomega{x}{y}{z} = \mtx{ 0 & \omeg{x}{y}{z}} \in \quatsv$,  $\vel{x}{y}{z} \in \reals{3}$} is the translational velocity of frame $\cf{x}$ relative to frame $\cf{y}$ expressed in frame $\cf{z}$ coordinates, \mbox{$\qvel{x}{y}{z} = \mtx{ 0 & \vel{x}{y}{z} } \in \quatsv$}, and $\pos{z}{x}{z} \in \reals{3}$ is the position of frame $\cf{z}$ relative to frame $\cf{x}$ expressed in frame $\cf{z}$ coordinates, \mbox{$\qpos{z}{x}{z} = \mtx{ 0 & \pos{z}{x}{z} } \in \quatsv$}. Calculating relative velocities between dual velocities expressed in the same coordinate frame follows standard vector subtraction \mbox{$\dualomega{x}{y}{z} = \dualomega{x}{w}{z} - \dualomega{y}{w}{z}.$}

The dual quaternion kinematics are
\begin{align}
    \ddualq{x}{y} = \frac{1}{2} \dualq{x}{y} \dualomega{x}{y}{x} = \frac{1}{2} \dualomega{x}{y}{y} \dualq{x}{y}.  \label{eq:dqkin}
\end{align}
The simplicity and resemblance to the quaternion kinematics in \eqref{eq:qkin} once more highlight dual quaternions as a natural extension of quaternions to \se.

\subsubsection*{Matrix Form}
Dual quaternion operations can also be expressed in matrix form, where a dual quaternion $\dualqblank{q} \in \dquats$ left multiplied by a matrix follows the standard matrix multiplication algebra as if $\dualqblank{q}$ were a $\reals{8 \times 1}$ vector:
\begin{sgather}
    \mblank{M} = \mtx{\mblank{M}_{1} & \mblank{M}_{2} \\ \mblank{M}_{3} & \mblank{M}_{4}} \in \reals{8 \times 8},  \\
    \mblank{M} \dualqblank{q} = (\mblank{M}_{1} \qreal{q} + \mblank{M}_{2} \qdual{q}) + \dualunit (\mblank{M}_{3} \qreal{q} + \mblank{M}_{4} \qdual{q}) \in \dquats,
\end{sgather}
where $\mblank{M}_{1}, \mblank{M}_{2}, \mblank{M}_{3}, \mblank{M}_{4} \in \reals{4 \times 4}$.

For a dual quaternion $\dualqblank{q} \in \dquats$, the left and right dual quaternion multiplication matrices are
\begin{salign}
    \qleft{\dualqblank{q}} &= \mtx{\qleft{\qreal{\qblank{q}}} & \zeros{4}{4} \\
    \qleft{\qdual{\qblank{q}}} & \qleft{\qreal{\qblank{q}}}} \in \reals{8 \times 8}  \\
    \qright{\dualqblank{q}} &= \mtx{\qright{\qreal{\qblank{q}}} & \zeros{4}{4} \\
    \qright{\qdual{\qblank{q}}} & \qright{\qreal{\qblank{q}}}} \in \reals{8 \times 8} \label{eq:dqleftright}. 
\end{salign}
Dual quaternion multiplication can be reconstructed in matrix form as
\begin{sgather}
    \dualqblank{a} \dualqblank{b} = \qleft{\dualqblank{a}} \dualqblank{b} = \qright{\dualqblank{b}} \dualqblank{a}, \\
    \dualqblank{a} \dualqblank{b} \dualqblank{c} = \qleft{\dualqblank{a} \dualqblank{b}} \dualqblank{c} = \qright{\dualqblank{b} \dualqblank{c}} \dualqblank{a} = \left( \qright{\dualqblank{b}} \dualqblank{a} \right) \dualqblank{c} = \dualqblank{a} \left( \qright{\dualqblank{c}} \dualqblank{b}\right)
\end{sgather}
for dual quaternions $\dualqblank{a}, \dualqblank{b}, \dualqblank{c} \in \dquats$. The dual quaternion cross product has a matrix form as well:
\begin{gather}
    \qcross{\dualqblank{q}} = \mtx{\qcross{\qreal{\qblank{q}}} & \zeros{4}{4} \\
    \qcross{{\qdual{\qblank{q}}}} & \qcross{\qreal{\qblank{q}}}} \in \reals{8 \times 8},  
\end{gather}
such that $\dualqblank{a} \cross \dualqblank{b} = \qcross{\dualqblank{a}} \dualqblank{b}$. 

Lastly, the dual quaternion conjugate is also antihomomorphic:
\begin{gather}
    \qconj{\left( \dualqblank{a} \dualqblank{b} \right)} = \qconj{\dualqblank{b}} \qconj{\dualqblank{a}} 
\end{gather}
and can be rewritten using the dual conjugate matrix \mbox{$\dualeyeconj = \blkdiag{\eyeconj, \eyeconj} \in \reals{8 \times 8}$},
\begin{gather}
    \qconj{\dualqblank{q}} = \eyeconj \dualqblank{q},
\end{gather}
where $\blkdiag{\cdot}$ is the block diagonal matrix function which returns a block diagonal matrix with its input matrices as the diagonal blocks and zero matrices on the off-diagonal blocks.

\subsubsection*{Derivatives}
By rewriting dual quaternion operations in a matrix form, dual quaternion derivatives can be calculated using matrix calculus tools and can be written in a compact form that mirrors quaternion derivatives. Some common derivatives that will be used in the following observability analysis sections are:
\begin{gather}
    \pd{\dualqblank{a} \dualqblank{b}}{\dualqblank{a}} = \qright{\dualqblank{b}},  \quad
    \pd{\dualqblank{a} \dualqblank{b}}{\dualqblank{b}} = \qleft{\dualqblank{a}}, \quad 
    \pd{\qconj{\dualqblank{a}} \dualqblank{b} \dualqblank{a}}{\dualqblank{a}} = \qleft{\qconj{\dualqblank{a}} \dualqblank{b}} + \qright{\dualqblank{b} \dualqblank{a}} \dualeyeconj \label{eq:dq_deriv}. 
\end{gather}

\subsection*{Rigid Body Dynamics} \label{sec:model}
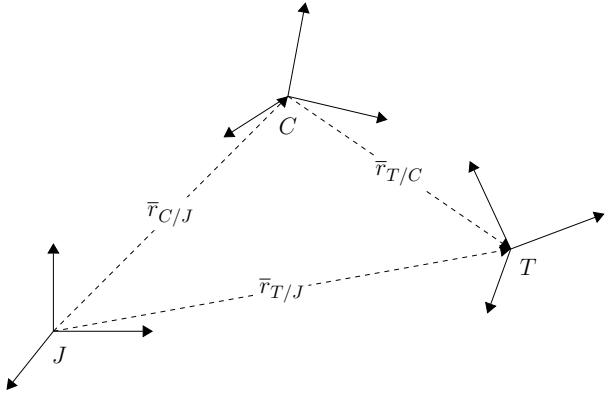
\begin{figure}[t]
\centering
\resizebox{.45\textwidth}{!}{%
\begin{circuitikz}
\tikzstyle{every node}=[font=\fontsize{15pt}{22pt}\selectfont]
\draw [line width=0.6pt, -{Triangle[scale=1.5]}, ] (3.625,5.75) -- (5.75,5.75);
\draw [line width=0.6pt, -{Triangle[scale=1.5]}, ] (3.625,5.75) -- (3.625,7.625);
\draw [line width=0.6pt, -{Triangle[scale=1.5]}, ] (3.625,5.75) -- (2.625,4.5);
\draw [line width=0.6pt, -{Triangle[scale=1.5]}, ] (8.625,10.75) -- (10.75,10.25);
\draw [line width=0.6pt, -{Triangle[scale=1.5]}, ] (8.625,10.75) -- (7.25,9.875);
\draw [line width=0.6pt, -{Triangle[scale=1.5]}, ] (13.375,7.5) -- (12.5,9.375);
\draw [line width=0.6pt, -{Triangle[scale=1.5]}, ] (13.375,7.5) -- (12.875,6.125);
\draw [line width=0.6pt, -{Triangle[scale=1.5]}, dashed] (3.625,5.75) -- (8.625,10.75)node[pos=0.5, fill={rgb,255:red,255; green,255; blue,255}, fill opacity=1, text opacity=1, inner xsep=0.080cm, inner ysep=0.085cm, rounded corners=0.020cm]{$\overline{r}_{C/J}$};
\node [font=\fontsize{15pt}{22pt}\selectfont, inner xsep=0.080cm, inner ysep=0.085cm, rounded corners=0.005cm] at (3.75,5.25) {$J$};
\node [font=\fontsize{15pt}{22pt}\selectfont, inner xsep=0.080cm, inner ysep=0.085cm, rounded corners=0.020cm] at (8.625,10.125) {$C$};
\node [font=\fontsize{15pt}{22pt}\selectfont, inner xsep=0.080cm, inner ysep=0.085cm, rounded corners=0.020cm] at (13.75,7.125) {$T$};
\draw [line width=0.6pt, -{Triangle[scale=1.5]}, ] (8.625,10.75) -- (9,12.75);
\draw [line width=0.6pt, -{Triangle[scale=1.5]}, ] (13.375,7.5) -- (15.375,8.25);
\draw [line width=0.6pt, -{Triangle[scale=1.5]}, dashed] (8.625,10.75) -- (13.375,7.5)node[pos=0.5, fill={rgb,255:red,255; green,255; blue,255}, fill opacity=1, text opacity=1, inner xsep=0.080cm, inner ysep=0.085cm, rounded corners=0.020cm]{$\overline{r}_{T/C}$};
\draw [line width=0.6pt, -{Triangle[scale=1.5]}, dashed] (3.625,5.75) -- (13.375,7.5)node[pos=0.5, fill={rgb,255:red,255; green,255; blue,255}, fill opacity=1, text opacity=1, inner xsep=0.080cm, inner ysep=0.085cm, rounded corners=0.020cm]{$\overline{r}_{T/J}$};
\end{circuitikz}
}%
\caption{Diagram of the target ($\target$), camera ($\camera$), and inertial ($\inertial$) coordinate frames and the relative position vectors \mbox{($\pos{\target}{\inertial}{}, \pos{\camera}{\inertial}{}, \pos{\target}{\camera}{} \in \reals{3}$)} between them.}
\label{fig:rel_frames}
\end{figure}
The 6-DOF relative motion between two rigid bodies can be modeled using dual quaternions to represent pose, translational velocity, and angular velocity. Motivated by visual target tracking, define two body-fixed frames, the target frame, $\target$, and the camera frame, $\camera$, along with an inertial reference frame, $\inertial$. \cref{fig:rel_frames} illustrates the geometry of these three coordinate frames. 

The net external forces, $\force{\camera} \in \reals{3}$, and torques, $\torque{\camera} \in \reals{3}$, applied at the camera center of mass in frame $\cf{\camera}$ coordinates can be lifted to vector quaternions and embedded into the force vector dual quaternion $\dualforce{\camera} \in \dquatsv$:
\begin{sgather}
    \qforce{\camera} = \mtx{0 & \force{\camera}} \in \quatsv, \quad \qtorque{\camera} = \mtx{ 0 & \torque{\camera} } \in \quatsv, \\
    \dualforce{\camera} = \qforce{\camera} + \dualunit \qtorque{\camera}.
\end{sgather}
The mass matrix, $\mass{\camera}$, contains the mass, $m$, and inertia matrix, $\inertia \in \mathbb{S}_{++}^3$, properties of the system in frame $\cf{\camera}$ coordinates
\begin{gather}
    \mass{\camera} = \begin{bmatrix}
        \zeros{4}{4} & 
        \begin{matrix}
            1 & \zeros{3}{1} \\
            \zeros{1}{3} & m \eye{3}
        \end{matrix} \\
        \begin{matrix}
            1 & \zeros{3}{1} \\
            \zeros{1}{3} & \inertia
        \end{matrix} & \zeros{4}{4}
    \end{bmatrix}. 
\end{gather}

From \eqref{eq:dqkin} and \cite{Wang2012-lx}, the equations of motion for the camera frame with respect to the inertial frame, $\inertial$, in dual quaternion form are 
\begin{salign}
    \ddualq{\camera}{\inertial} &= \frac{1}{2} \dualq{\camera}{\inertial} \dualomega{\camera}{\inertial}{\camera},  \\
    \ddualomega{\camera}{\inertial}{\camera} &= \inv{\left(\mass{\camera}\right)}  \left(\dualforce{\camera} - \dualomega{\camera}{\inertial}{\camera} \cross \mass{\camera} \dualomega{\camera}{\inertial}{\camera}\right).  \label{eq:ci_dynamics}
\end{salign}
The target frame with respect to the inertial frame has the same form:
\begin{salign}
    \ddualq{\target}{\inertial} &= \frac{1}{2} \dualq{\target}{\inertial} \dualomega{\target}{\inertial}{\target},  \\
    \ddualomega{\target}{\inertial}{\target} &= \inv{\left(\mass{\target}\right)}  \left(\dualforce{\target} - \dualomega{\target}{\inertial}{\target} \cross \mass{\target} \dualomega{\target}{\inertial}{\target} \right). 
\end{salign}

The relative motion dynamics are derived by first differentiating $\dualomega{\target}{\camera}{\camera} = \dualomega{\target}{\inertial}{\camera} - \dualomega{\camera}{\inertial}{\camera}$ with respect to time to yield
\begin{gather}
    \ddualomega{\target}{\camera}{\camera} = \ddualomega{\target}{\inertial}{\camera} - \ddualomega{\camera}{\inertial}{\camera}. \label{eq:relomega} 
\end{gather}
Next, applying Proposition 1 from \cite{Filipe2013-bj} to derive an equivalent expression for $\ddualomega{\target}{\inertial}{\camera}$ and simplifying gives
\begin{align}
    \ddualomega{\target}{\inertial}{\camera} &= \dualq{\target}{\camera} \left( \ddualomega{\target}{\inertial}{\target} + \dualomega{\target}{\camera}{\target} \cross \dualomega{\target}{\inertial}{\target} \right) \qconj{\dualq{\target}{\camera}} \\
    &= \dualq{\target}{\camera} \ddualomega{\target}{\inertial}{\target} \qconj{\dualq{\target}{\camera}} + \dualomega{\target}{\camera}{\camera} \cross \left( \dualomega{\target}{\camera}{\camera} + \dualomega{\camera}{\inertial}{\camera} \right) \\
    &= \dualq{\target}{\camera} \ddualomega{\target}{\inertial}{\target} \qconj{\dualq{\target}{\camera}} + \dualomega{\target}{\camera}{\camera} \cross \dualomega{\camera}{\inertial}{\camera} \label{eq:omega_dot}.
\end{align}
The expression for $\ddualomega{\target}{\inertial}{\camera}$ in \eqref{eq:omega_dot} is then substituted into \eqref{eq:relomega} to yield the complete relative pose dynamics
\begin{subequations} \label{eq:eom}
    \begin{align}
    \ddualq{\target}{\camera} &= \frac{1}{2}  \dualomega{\target}{\camera}{\camera} \dualq{\target}{\camera}, \label{eq:eom_kin} \\
    \ddualomega{\target}{\camera}{\camera} &= \dualq{\target}{\camera} \ddualomega{\target}{\inertial}{\target} \qconj{\dualq{\target}{\camera}} 
    + \dualomega{\target}{\camera}{\camera} \cross \dualomega{\camera}{\inertial}{\camera} 
    - \ddualomega{\camera}{\inertial}{\camera}.
    \end{align}
\end{subequations}

To maintain generality, no assumptions are made regarding any particular model for the external forces and torques. 
For a concrete example of such modeling, see \cite{Filipe2015-zc} for a dual quaternion representation of forces and torques acting on a spacecraft.

\section{Nonlinear Observability} \label{sec:nlobsv}
Below, a brief review {is provided} of nonlinear observability and the Lie algebraic approach for determining the observability of nonlinear systems, summarized from \cite{Hermann1977-rt, Nijmeijer1990, Mirzaei2008-kl}.

Consider the nonlinear system, \(\Sigma\), with the following process and measurement models:
\begin{salign}
    \smash{\raisebox{-0.5\baselineskip}{$\Sigma:\quad$}} \dot{\state} &= f(\state, \control), \\
    \phantom{\Sigma:\quad} \meas &= h(\state),
\end{salign}
where \( \state(t) \in \mathbb{R}^{n} \), \( \meas(t) \in \mathbb{R}^{k} \), \( \control(t) \in \mathcal{U} \subseteq \mathbb{R}^{m} \), and \(\mathcal{U}\) is the set of admissible controls. Let \( \state(t, \state_0, \control) \) denote the solution trajectory to the initial value problem for \(\Sigma\) with initial condition \( \state(0) = \state_0 \) under the control input \( \control(t) \), and define \( \meas(t, \state_0, \control) = h(\state(t, \state_0, \control)) \) to be the measurement signal over the solution trajectory.
Initial states \( \state_A \) and \( \state_B \) at time $t=0$ are \mbox{\( V \)-indistinguishable} if for every control \( \control \in \mathcal{U} \), the corresponding state trajectories \( \state(t, \state_A, \control) \) and \( \state(t, \state_B, \control) \) remain in \( V \subseteq \mathbb{R}^n \) over \( t \in [0, T] \) and satisfy  
\begin{equation}
    \meas(t, \state_A, \control) = \meas(t, \state_B, \control) \ \forall \ t \in [0, T]. 
\end{equation}

The system \(\Sigma\) is {locally observable at \( \state_A \)} if there exists a neighborhood \( W \) of \( \state_A \) such that, for every neighborhood \( V \subset W \), indistinguishability within \( V \) implies \( \state_A = \state_B \). If \(\Sigma\) is locally observable at all \( \state \), then the system is deemed {locally observable}. Intuitively, this definition means that \( \state_A \) can be distinguished from nearby states within finite time and with trajectories remaining close to \( \state_A \). 

A differential geometric approach to determining local observability involves analyzing the {Lie derivatives} of the output function \( h(\state) \) with respect to the process model \( f(\state, \control) \). The zeroth through second-order Lie derivatives are given by:
\begin{salign}
    \lie{h}{f}{0} &= h(\state),  \\
    \lie{h}{f}{1} &= \nabla h(\state) \cdot f(\state, \control),  \\
    \lie{h}{f}{2} &= \lie{\left( \lie{h}{f}{1} \right)}{f}{1} = \nabla (\lie{h}{f}{1}) \cdot f(\state, \control) . 
\end{salign}
Higher-order Lie derivatives follow a similar recursive form, and derivatives are computed with respect to $\state$. If \( h(\state) \) is a scalar function, \( \nabla h(\state) \) is the gradient and is represented as a row vector in this work. If \( h(\state) \) is a vector function, \( \nabla h(\state) \) is the Jacobian matrix.  

If the process model is {control-affine}, meaning it decomposes as: $\dot{\state} = f_0(\state) + \sum_{i=1}^{m} f_i(\state) u_i$, where \( f_0(\state) \) is the drift vector field and $f_i(\state)$ are the control vector fields, then Lie derivatives can be computed separately with respect to each of the vector fields. For instance, a second-order mixed Lie derivative of \( h(\state) \), first along \( f_0(\state) \) and then along \( f_1(\state) \), is:
\begin{equation}
    \mathcal{L}_{f_1 f_0}^2 h = \mathcal{L}_{f_1}^1 (\mathcal{L}_{f_0}^1 h) = \nabla (\mathcal{L}_{f_0}^1 h) \cdot f_1(\state).
\end{equation}
Taking a Lie derivative with respect to \( f_i(\state) \) implicitly assumes \( u_i \not\equiv 0 \).

The {observability Lie algebra} $\obsvspace{}$ for a nonlinear control-affine system is defined as
\begin{equation}
    \obsvspace{} = \text{span}\left\{ \mathcal{L}_{f_i}^d h(\state) \mid d \in \mathbb{N}_0, \; i = 0, \dots, m \right\}.
\end{equation}
In this definition the Lie derivatives can be any combination of the process model vector fields and of arbitrary degree. Local observability is then determined using the following rank condition on the Jacobian of $\obsvspace{}${, $\obsv{}$}:
\begin{theorem}\cite{Hermann1977-rt}
    The nonlinear system is locally observable if \( \rank{\obsv{}} = n \) .
\end{theorem}

In linear systems, $\obsv{}$ reduces to the classical observability matrix, and its dimensions are determined by the dimensionality of the system. However, for nonlinear systems, while the number of columns of $\obsv{}$ is determined by the state dimension $n$, the number of rows necessary to satisfy the full rank condition is not known a priori. While no universal method exists for constructing $\obsv{}$, sequentially computing low-order Lie derivatives along well-chosen combinations of the process model vector fields typically yields good results. If any set of Lie derivatives generates a $\obsv{}$ that satisfies the rank condition, then the system is locally observable. In practice, if a nonlinear system is locally observable, then a state estimator, such as a Kalman filter, will have its estimate error covariance converge towards the Cramér--Rao lower bound, and estimation errors will remain statistically bounded by the estimate error covariance when initialized with a state estimate sufficiently close to the true state.

\subsection{State Definition}
In this work, the state $\state \in \dquatsu \times \dquatsv$, is composed of the dual pose and dual velocity quaternions of frame $\target$ with respect to frame $\camera$ in frame $\camera$ coordinates (from \eqref{eq:kinxy} and \eqref{eq:dqvel}):
\begin{align} \label{eq:statedef}
    \state &= \mtx{\dualq{\target}{\camera} \\ \dualomega{\target}{\camera}{\camera}} 
    = \mtx{\q{\target}{\camera} + \epsilon
            \frac{1}{2} \qpos{\target}{\camera}{\camera} \q{\target}{\camera} \\
            \qomega{\target}{\camera}{\camera}  +\epsilon\left(
            \qvel{\target}{\camera}{\camera} - \qomega{\target}{\camera}{\camera} \cross \qpos{\target}{\camera}{\camera} \right)}.
\end{align}

To simplify the observability analyses in the following sections, the following state transformation {is applied}, which is a decomposition of the dual quaternion states into their real and dual components. Let 
$\statetrans$ be the transformed quaternion state of \eqref{eq:statedef}:
\begin{align} \label{eq:transformed-state}
    \statetrans
    = \mtx{\q{\target}{\camera} \\
            \qmu{\target}{\camera} \\
            \qomega{\target}{\camera}{\camera} \\
            \qbeta{\target}{\camera}{\camera}} \in \quatsu \times \quats \times \quatsv \times \quatsv,
\end{align}
where $\qmu{\target}{\camera} = \qdual{(\dualq{\target}{\camera})} = \frac{1}{2} \qpos{\target}{\camera}{\camera} \q{\target}{\camera}$ and \mbox{$\qbeta{\target}{\camera}{\camera} = \qdual{(\dualomega{\target}{\camera}{\camera})} = \qvel{\target}{\camera}{\camera} - \qomega{\target}{\camera}{\camera} \cross \qpos{\target}{\camera}{\camera}$}. Formally, since a bijective transformation exists between $\state$ and $\statetrans$, observability of $\statetrans$ is equivalent to observability of $\state$ \cite{Andrews2024-jk}. The subsequent observability analyses therefore focus on establishing observability of $\statetrans$.

Recall from \eqref{eq:eom_kin}, the dual quaternion kinematic equation is \mbox{$\ddualq{\target}{\camera} = \frac{1}{2} \dualomega{\target}{\camera}{\camera} \dualq{\target}{\camera}$}. Expressing the dual components of $\dualq{\target}{\camera}$ and $\dualomega{\target}{\camera}{\camera}$ in terms of the transformed state variables from \eqref{eq:transformed-state}, and evaluating the dual quaternion kinematic equation yields the transformed state kinematic equations
\begin{salign}
    \dq{\target}{\camera} &= \frac{1}{2} \qomega{\target}{\camera}{\camera} \q{\target}{\camera}, \\
    \dqmu{\target}{\camera} &= \frac{1}{2} \qbeta{\target}{\camera}{\camera} \q{\target}{\camera} + \qomega{\target}{\camera}{\camera} \qmu{\target}{\camera},
\end{salign}
where $\dq{\target}{\camera} = \qreal{(\ddualq{\target}{\camera})}$ and $\dqmu{\target}{\camera} = \qdual{(\ddualq{\target}{\camera})}$.

\section{Position Vector Measurements} \label{sec:position_analysis}
The first measurement model used for observability analysis is based on a series of relative position measurements. This model captures several practical measurement types, including cameras measuring fiducial markers or radio transponders that return range and bearing data. For notational consistency and to build on the visual target tracking application introduced in previous sections, the measurement model is presented in the context of a camera measuring fiducial markers. In this model it is assumed that the position of each marker with respect to the target coordinate frame, $\pos{\marker}{\target}{\target} \in \reals{3}$, is known to the observer, where $\marker$ denotes the marker coordinate frame and $\qpos{\marker}{\target}{\target} = \mtx{0 & \pos{\marker}{\target}{\target}} \in \quatsv$.

The position vector measurement function for a single marker is
\begin{align}
    \meas &= h(\state) = \qpos{\marker}{\camera}{\camera} = \qpos{\target}{\camera}{\camera} + \q{\target}{\camera} \qpos{\marker}{\target}{\target} \qconj{\q{\target}{\camera}} \in \quatsv,
\end{align}
where $\q{\target}{\camera} \in \quatsu$ is the orientation of the target frame $\target$ with respect to the camera frame $\camera$, and $\qpos{\target}{\camera}{\camera} \in \quatsv$ is the position of frame $\target$ with respect to frame $\camera$ in frame $\camera$ coordinates. Substituting in the transformed state, the measurement function is rewritten as
\begin{align} \label{eq:meas_transformed}
    \meas = h(\statetrans) &= 2 \qmu{\target}{\camera} \qconj{\q{\target}{\camera}} + \q{\target}{\camera} \qpos{\marker}{\target}{\target} \qconj{\q{\target}{\camera}}.
\end{align}

In Section 4.3.2 of \cite{Sola2017-mp} an explicit form is presented for the derivative $\pd{\qblank{q} \qblank{a} \qconj{\qblank{q}}}{\qblank{q}}$ where $\qblank{q} \in \quatsu$ and $\qblank{a} \in \quatsv$. In this work, the alternative form \eqref{eq:qderiv} {is utilized} which is instead written in terms of the left and right quaternion multiplication matrices and will prove advantageous for matrix rank analysis:
\begin{align}
    \pd{\qblank{q} \qblank{a} \qconj{\qblank{q}}}{\qblank{q}} &= \qright{\qblank{a} \qconj{\qblank{q}}} + \qleft{\qblank{q} \qblank{a}} \eyeconj.
\end{align}
For brevity, let $\qblank{b} = \qblank{q} \qblank{a} \in \quats$, {note that then} $-\qconj{\qblank{b}} = \qblank{a} \qconj{\qblank{q}}$, and define the function
\begin{align}
    2 \kfunc{\qblank{b}} &= \qright{-\qconj{\qblank{b}}} + \qleft{\qblank{b}} \eyeconj \\
    &= 2 \mtx{0 & \zeros{1}{3} \\
    \qvec{\qblank{b}} & -\qscalar{\qblank{b}} \eye{3} - \skewmat{\qblank{b}}}.  \label{eq:Kb}
\end{align}

The rank properties of $\kfunc{\qblank{q} \qblank{a}}$ are proven in the following lemma.
\begin{lemma} \label{lemma:k_mat}
    For $\qblank{q} \in \quatsu$ and $\qblank{a} \in \quatsv$,
    \begin{gather}
        \rank{\kfunc{\qblank{{\qblank{q} \qblank{a}}}}} =
        \begin{cases}
        0 & \text{if } \qblank{{a}} = \qzero \\
        3 & \text{otherwise.}
        \end{cases}
    \end{gather}
\end{lemma}
\begin{proof}
{First, rewrite $\qzero = \qblank{q} \qblank{a}$ as $\qzero = \qleft{\qblank{q}} \qblank{a}$. Applying Lemma 3 from \cite{Andrews2024-jk} gives that $\qleft{\qblank{q}}$ is full rank if $\qblank{q} \in \quatsu$.  Therefore, the solution to the equation $\qblank{b} = \qzero = \qblank{q} \qblank{a} $ has only the trivial solution $\qblank{a} = \qzero$, for which $\rank{\kfunc{\qblank{{\qblank{q} \qblank{a}}}}} = 0$.}

If $a \neq 0$, then because $\qleft{\qblank{q}}$ is full rank, $b=qa$ has at least one non-zero element.
From \eqref{eq:Kb}, the elementwise form of $\kfunc{\qblank{{b}}}$ is
\begin{gather}
    \kfunc{\qblank{{b}}} =
    \mtx{
    0 & 0 & 0 & 0 \\
    b_1 & -b_0 & b_3 & -b_2 \\
    b_2 & -b_3 & -b_0 & b_1 \\
    b_3 & b_2 & -b_1 & -b_0
    }.
\end{gather}
{If $b_i \neq 0$ for some $i \in \{0, 1, 2, 3\}$, then by inspection {one} can see that a linear transformation of $\kfunc{\qblank{b}}$ {exists} such that there will be a $3 \times 3$ block matrix of the form \mbox{$b_i \eye{3} + \mathbf{S}_\times$} where \mbox{$\mathbf{S}_\times \in \reals{3 \times 3}$} is a skew-symmetric matrix. A non-zero constant times the identity matrix plus a skew-symmetric matrix is known to be full rank, hence $\rank{\kfunc{\qblank{b}}} = 3$ if $\qblank{a} \neq \qzero$.}
\end{proof}

\subsection*{Zeroth Order Lie Derivative}
The zeroth order Lie derivative {contribution to the Jacobian of the observability Lie algebra} is the derivative of the measurement function with respect to the transformed state variables\kamedit{:}
\begin{align}
    \lie{h}{f}{0} &= h(\statetrans) = 2 \qmu{\target}{\camera} \qconj{\q{\target}{\camera}} + \q{\target}{\camera} \qpos{\marker}{\target}{\target} \qconj{\q{\target}{\camera}}, 
\end{align}
{The Jacobian of this Lie derivative is found by direct computation as follows:}
\begin{align}
    \jacobian \lie{h}{f}{0} &= \mtx{\pd{\lie{h}{f}{0}}{\q{\target}{\camera}} & \pd{\lie{h}{f}{0}}{\qmu{\target}{\camera}} & \pd{\lie{h}{f}{0}}{\qomega{\target}{\camera}{\camera}} & \pd{\lie{h}{f}{0}}{\qbeta{\target}{\camera}{\camera}}}  \\
    &= \mtx{\pd{\lie{h}{f}{0}}{\q{\target}{\camera}} & \pd{\lie{h}{f}{0}}{\qmu{\target}{\camera}} & \zeros{4}{4} & \zeros{4}{4}}, \\
    \pd{\lie{h}{f}{0}}{\q{\target}{\camera}} &= \qright{\qpos{\marker}{\target}{\target} \qconj{\q{\target}{\camera}}} + \qleft{\q{\target}{\camera} \qpos{\marker}{\target}{\target}} \eyeconj \\
    & ~~~~ + 2 \qleft{\qmu{\target}{\camera}} \eyeconj  \\
    &= 2 \left( \kfunc{\q{\target}{\camera} \qpos{\marker}{\target}{\target}} + \qleft{\qmu{\target}{\camera}} \eyeconj \right),  \\
    \pd{\lie{h}{f}{0}}{\qmu{\target}{\camera}} &= 2 \qright{\qconj{\q{\target}{\camera}}}
\end{align}

\subsection*{First Order Lie Derivative}
The first order Lie derivative is calculated and factored in terms of the zeroth order Lie derivative:
\begin{align}
    \lie{h}{f}{1} &= \jacobian \lie{h}{f}{0} \cdot f(\statetrans, 
    \control)  \\
    &= \pd{\lie{h}{f}{0}}{\q{\target}{\camera}} \dq{\target}{\camera} + \pd{\lie{h}{f}{0}}{\qmu{\target}{\camera}} \dqmu{\target}{\camera}  \\
    &= \frac{1}{2} \pd{\lie{h}{f}{0}}{\q{\target}{\camera}} \qomega{\target}{\camera}{\camera} \q{\target}{\camera}  \\
     &~~~~ + \frac{1}{2} \pd{\lie{h}{f}{0}}{\qmu{\target}{\camera}} \left( \qbeta{\target}{\camera}{\camera} \q{\target}{\camera}
    + \qomega{\target}{\camera}{\camera} \qmu{\target}{\camera} \right).
    \end{align}
{The Jacobian of this expression is then}
\begin{salign}
    \jacobian \lie{h}{f}{1} &= \mtx{\pd{\lie{h}{f}{1}}{\q{\target}{\camera}} & \pd{\lie{h}{f}{1}}{\qmu{\target}{\camera}} & \pd{\lie{h}{f}{1}}{\qomega{\target}{\camera}{\camera}} & \pd{\lie{h}{f}{1}}{\qbeta{\target}{\camera}{\camera}}},  \\
    \pd{\lie{h}{f}{1}}{\qomega{\target}{\camera}{\camera}} &= \frac{1}{2} \left( \pd{\lie{h}{f}{0}}{\q{\target}{\camera}} \qright{\q{\target}{\camera}} + \pd{\lie{h}{f}{0}}{\qmu{\target}{\camera}} \qright{\qmu{\target}{\camera}} \right), \label{eq:lie1} \\
    \pd{\lie{h}{f}{1}}{\qbeta{\target}{\camera}{\camera}} &= \frac{1}{2} \pd{\lie{h}{f}{0}}{\qmu{\target}{\camera}} \qright{\q{\target}{\camera}} \label{eq:lie2}. 
\end{salign}
The terms $\pd{\lie{h}{f}{1}}{\q{\target}{\camera}}$ and $\pd{\lie{h}{f}{1}}{\qmu{\target}{\camera}}$ are omitted because they are analytically intractable and are not necessary in the observability analysis. Interestingly, \eqref{eq:lie1} and \eqref{eq:lie2} can be factored and written as the matrix product
\begin{align}
    \mtx{\pd{\lie{h}{f}{1}}{\qomega{\target}{\camera}{\camera}} & \pd{\lie{h}{f}{1}}{\qbeta{\target}{\camera}{\camera}}} & = \frac{1}{2} \mtx{\pd{\lie{h}{f}{0}}{\q{\target}{\camera}} & \pd{\lie{h}{f}{0}}{\qmu{\target}{\camera}}} \qright{\dualq{\target}{\camera}}.
\end{align}


\subsection*{Observability Codistribution}
For the observability analysis, three independent markers, denoted as coordinate frames $\marker1, \marker2, \marker3$, are assumed to be simultaneously visible and to have non-zero range from the observing camera. The concatenated measurement function is
\begin{align}
    \meas &= h(\statetrans) = \mtx{ \qpos{\marker1}{\camera}{\camera} \\  \qpos{\marker2}{\camera}{\camera} \\  \qpos{\marker3}{\camera}{\camera}} \in \quatsv \times \quatsv \times \quatsv.
\end{align}
To express the Jacobian of the zeroth order Lie derivative of the concatenated measurement function, define the matrix \mbox{$\dmat \in \reals{12 \times 8}$} as
\begin{align} \label{eq:fmat}
    \dmat &= \mtx{\kfunc{\q{\target}{\camera} \qpos{\marker1}{\target}{\target}} + \qleft{\qmu{\target}{\camera}} \eyeconj & \qright{\qconj{\q{\target}{\camera}}} \\
    \kfunc{\q{\target}{\camera} \qpos{\marker2}{\target}{\target}} + \qleft{\qmu{\target}{\camera}} \eyeconj & \qright{\qconj{\q{\target}{\camera}}} \\
    \kfunc{\q{\target}{\camera} \qpos{\marker3}{\target}{\target}} + \qleft{\qmu{\target}{\camera}} \eyeconj & \qright{\qconj{\q{\target}{\camera}}}} .
\end{align}
The observability codistribution matrix consisting of the zeroth and first order Lie derivatives is then 
\begin{align} \label{eq:sys2_obsv}
    \obsv{} &= \mtx{2 \dmat & \zeros{12}{8} \\
    \star_{12 \times 8} & \dmat \qright{\dualq{\target}{\camera}}} \in \reals{24 \times 16}.
\end{align}
The symbol $\star$ denotes non-zero derivatives whose full structure is omitted for brevity as they will not be needed for the desired observability result.

The conditions for observability are now formally stated.
\begin{theorem} \label{thm:obsv_pos}
    The observability codistribution, $\obsv{}$, is full rank and the nonlinear system \eqref{eq:statedef} is locally observable with three relative position measurements if the non-collinear condition is satisfied: \mbox{$\qpos{\marker2}{\marker1}{\target} \cross \qpos{\marker3}{\marker1}{\target} \neq \qzero$}.
\end{theorem}
\begin{proof} 
    Performing block row reduction on $\dmat$ by subtracting the first row from the second and third, and simplifying using the linearity property $\kfunc{\qblank{q} \qblank{a}} - \kfunc{\qblank{q} \qblank{b}} = \kfunc{\qblank{q} \left(\qblank{a} - \qblank{b}\right)}$ for $\qblank{a}, \qblank{b} \in \quats$ and $\qblank{q} \in \quatsu$ yields
    \begin{gather}
        \mtx{\kfunc{\q{\target}{\camera} \qpos{\marker1}{\target}{\target}} & \qright{\qconj{\q{\target}{\camera}}} \\
        \kfunc{\q{\target}{\camera} \qpos{\marker2}{\marker1}{\target}} & \zeros{4}{4} \\
        \kfunc{\q{\target}{\camera} \qpos{\marker3}{\marker1}{\target}} & \zeros{4}{4}}.
    \end{gather}
    The first block row has rank four because \mbox{$\rank{\qright{\qconj{\q{\target}{\camera}}}} = 4$ \cite{Andrews2024-jk}}. Applying Lemma \ref{lemma:k_mat}, the second and third block rows each have rank three provided that $\qpos{\marker2}{\marker1}{\target} \neq \qzero$ and $\qpos{\marker3}{\marker1}{\target} \neq \qzero$. The first block row is clearly linearly independent from the second and third. The second and third block rows are linearly independent from each other if $\qpos{\marker2}{\marker1}{\target} \neq \kappa \, \qpos{\marker3}{\marker1}{\target}$ for any non-zero scalar $\kappa$. Since the block rows are mutually linearly independent, their ranks sum to yield $\rank{\dmat} = 8$.

    The matrix $\obsv{}$ \eqref{eq:sys2_obsv} has a block triangular structure, with its lower diagonal block equal to the upper diagonal block right-multiplied by $\frac{1}{2} \qright{\dualq{\target}{\camera}}$. Since $\qright{\dualq{\target}{\camera}}$ is always full rank, 
    $\obsv{}$ achieves full rank if $\dmat$ has full rank because $\obsv{}$ then becomes a block triangular matrix with full-rank diagonal blocks \cite{blockrank}.
\end{proof}

Intuitively, the observability condition requires that the three marker locations form a plane rather than a line. This geometric constraint is necessary because collinear markers would fail to provide rotational information about all axes.

\begin{table*}[t]
    \centering
    \begin{align} \label{eq:sys3_delta}
        \omat = \pmat \lmat = \mtx{\qfunc{\q{\target}{\camera}} & \zeros{4}{4} & \zeros{4}{4} & \zeros{4}{4} \\
        \zeros{4}{4} & \jacobian \ph{\qpos{\marker1}{\camera}{\camera}} & \zeros{4}{4} & \zeros{4}{4} \\
        \zeros{4}{4} & \zeros{4}{4} & \jacobian \ph{\qpos{\marker2}{\camera}{\camera}} & \zeros{4}{4} \\
        \zeros{4}{4} & \zeros{4}{4} & \zeros{4}{4} & \jacobian \ph{\qpos{\marker3}{\camera}{\camera}}}
        \mtx{\eye{4} & \zeros{4}{4} \\
        \kfunc{\q{\target}{\camera} \qpos{\marker1}{\target}{\target}} + \qleft{\qmu{\target}{\camera}} \eyeconj & \qright{\qconj{\q{\target}{\camera}}} \\
        \kfunc{\q{\target}{\camera} \qpos{\marker2}{\target}{\target}} + \qleft{\qmu{\target}{\camera}} \eyeconj & \qright{\qconj{\q{\target}{\camera}}} \\
        \kfunc{\q{\target}{\camera} \qpos{\marker3}{\target}{\target}} + \qleft{\qmu{\target}{\camera}} \eyeconj & \qright{\qconj{\q{\target}{\camera}}}} \in \reals{16 \times 8}
    \end{align}
\end{table*}
\section{Unit Vector Measurements} \label{sec:proj_analysis}
The next observability analysis considers the unit vector measurement model, which is functionally equivalent to common projected-vector measurement models including pinhole cameras, bearing-only radar, and angle-only star trackers. While these models differ in their specific derivations, they share a fundamental structure: each represents relative position vectors normalized by a distance metric. Consequently, although the Jacobian matrices of these measurement functions take different forms, they yield equivalent rank conditions and thus identical observability criteria. The unit vector representation is adopted here because it simplifies the mathematical analysis. 

To prove observability, the transformed state $\statetrans$ {is again leveraged} to simplify the derivatives. The unit vector measurement function, $\ph{\qblank{r}} \in \quatsv$, with input position vector $\qvec{r} \in \reals{3}$ and \mbox{$\qblank{r} = \mtx{0 & \qvec{r}} \in \quatsv$}, is defined as
\begin{gather}
    \ph{\qblank{r}} = 
    \frac{\qblank{r}}{\norm{\qblank{r}}}.
\end{gather}
The positive depth condition, $\norm{\qblank{r}} > 0$, must be satisfied for a valid measurement. In addition to the unit vector measurements, the unit quaternion constraint $\qconj{\q{\target}{\camera}} \q{\target}{\camera} = \qone$ {is included} as a pseudo-measurement. Incorporating the unit quaternion constraint as a pseudo-measurement is a common technique in filter design that enforces the geometric constraint during measurement updates, preventing the quaternion state from drifting off the unit sphere due to numerical errors and linearization approximations. This approach was similarly employed in the observability analysis in \cite{Mirzaei2008-kl} to maintain the unit norm constraint on quaternion states.

The measurement function, incorporating both the unit quaternion constraint and a single unit vector measurement, is
\begin{align} \label{eq-unitquat}
    \meas &= h(\statetrans) = \mtx{\qconj{\q{\target}{\camera}} \q{\target}{\camera} - \qone \\ 
    \ph{\qpos{\marker}{\camera}{\camera}}} \in \qzero \times \quatsv.
\end{align}

\subsection*{Zeroth Order Lie Derivative}
The derivative of a quaternion $\qblank{a} \in \quats$ multiplied by its conjugate is
\begin{align}
    \pd{\qconj{\qblank{a}} \qblank{a}}{\qblank{a}} &= \qright{\qblank{a}} \eyeconj + \qleft{\qconj{a}} \in \reals{4 \times 4}  \\
    &= 2 \mtx{\qblank{a} \\ \zeros{3}{4}} \\
    &\triangleq 2 \qfunc{\qblank{a}}.
\end{align}
To calculate the Jacobian of the unit vector function with respect to each of the transformed state variables, apply the chain rule
\begin{gather}
    \pd{\ph{\qpos{\marker}{\camera}{\camera}}}{\statetrans} = \pd{\ph{\qpos{\marker}{\camera}{\camera}}}{\qpos{\marker}{\camera}{\camera}} \pd{\qpos{\marker}{\camera}{\camera}}{\statetrans}.
\end{gather}
Note that $\pd{\qpos{\marker}{\camera}{\camera}}{\statetrans}$ is the Jacobian of the position vector measurement, whose form and properties were established in Section \ref{sec:position_analysis}.

Now taking the Lie derivative {of the measurement \eqref{eq-unitquat} and writing} $\jacobian \ph{\qpos{\marker}{\camera}{\camera}} = \pd{\ph{\qpos{\marker}{\camera}{\camera}}}{\qpos{\marker}{\camera}{\camera}}$ for brevity,
\begin{align}
    \lie{h}{f}{0} &= \mtx{\qconj{\q{\target}{\camera}} \q{\target}{\camera} - \qone \\ 
    \ph{\qpos{\marker}{\camera}{\camera}}}, \\
    \jacobian \lie{h}{f}{0} &= \mtx{\pd{\lie{h}{f}{0}}{\q{\target}{\camera}} & \pd{\lie{h}{f}{0}}{\qmu{\target}{\camera}} & \zeros{8}{4} & \zeros{8}{4}}, \\
    \pd{\lie{h}{f}{0}}{\q{\target}{\camera}} &= 2 \mtx{\qfunc{\q{\target}{\camera}} \\ \jacobian \ph{\qpos{\marker}{\camera}{\camera}} \left( \kfunc{\q{\target}{\camera} \qpos{\marker}{\target}{\target}} + \qleft{\qmu{\target}{\camera}} \eyeconj \right)}, \\
    \pd{\lie{h}{f}{0}}{\qmu{\target}{\camera}} &= 2 \mtx{\zeros{4}{4} \\ \jacobian \ph{\qpos{\marker}{\camera}{\camera}} \qright{\qconj{\q{\target}{\camera}}}}.
\end{align}

\subsection*{First Order Lie Derivative}
Using the same factorization approach as in Section \ref{sec:position_analysis}, the first-order Lie derivative Jacobian can be factored and rewritten in terms of the zeroth-order Lie derivative Jacobian:
\begin{align}
    \mtx{\pd{\lie{h}{f}{1}}{\qomega{\target}{\camera}{\camera}} & \pd{\lie{h}{f}{1}}{\qbeta{\target}{\camera}{\camera}}} = \frac{1}{2} \mtx{\pd{\lie{h}{f}{0}}{\q{\target}{\camera}} & \pd{\lie{h}{f}{0}}{\qmu{\target}{\camera}}} \qright{\dualq{\target}{\camera}}.
\end{align}

\subsection*{Observability Codistribution}
Now assume that three markers are simultaneously visible for three unit vector measurements. The concatenated measurement function is
\begin{align}
    \meas &= h(\statetrans) = \mtx{\qconj{\q{\target}{\camera}} \q{\target}{\camera} - \qone \\ 
    \ph{\qpos{\marker1}{\camera}{\camera}} \\
    \ph{\qpos{\marker2}{\camera}{\camera}} \\
    \ph{\qpos{\marker3}{\camera}{\camera}}} \in \qzero \times \quatsv \times \quatsv \times \quatsv.
\end{align}
The observability codistribution with zeroth and first order Lie derivatives is then
\begin{align} \label{eq:dounit}
    \obsv{} &= \mtx{2 \omat & \zeros{16}{8} \\
    \star_{16 \times 8} & \omat \qright{\dualq{\target}{\camera}}} \in \reals{32 \times 16},
\end{align}
where the block matrix $\omat$ is defined in \eqref{eq:sys3_delta}. The matrix $\omat$ is decomposed as $\omat = \pmat \lmat$, where $\pmat \in \reals{16 \times 16}$ is block diagonal, and $\lmat \in \reals{16 \times 8}$ extends $\dmat$ \eqref{eq:fmat} with an additional block row corresponding to the unit quaternion constraint pseudo-measurement. 

First, the following lemma will be used to derive the observability conditions.
\begin{lemma} \label{lem:rank}
    Let \( \mblank{A} \in \reals{m \times m} \) and \( \mblank{B} \in \reals{m \times n} \) with \( \mblank{B} \) full rank and overdetermined (i.e., \( m > n \), \(\rank{\mblank{B}} = n\)).  
    If \( \left\{ \vblank{n}_1, \vblank{n}_2, \dots, \vblank{n}_k \right\} \) is a basis for \( \nullspace{\mblank{A}} \) and for each \( \vblank{n}_i \) the system
    \begin{align}
        \mblank{B} \vblank{z} = \vblank{n}_i
    \end{align}
    has no solution, then \( \nullspace{\mblank{A} \mblank{B}} = \left\{ \zeros{n}{} \right\} \), and hence \( \mblank{A} \mblank{B} \) is full rank (i.e., \( \rank{\mblank{A} \mblank{B}} = n \)).
\end{lemma}

\begin{proof}
    First, observe that
    \begin{align}
        \nullspace{\mblank{A} \mblank{B}} &= \left\{ \vblank{z} \in \reals{n} \ \middle| \ \mblank{B} \vblank{z} \in \nullspace{\mblank{A}} \right\}
    \end{align}
    because $\mblank{B}$ is full rank. Suppose there exists a nonzero vector \( \vblank{z} \) such that \( \mblank{B} \vblank{z} \in \nullspace{\mblank{A}} \), then since \( \nullspace{\mblank{A}} \) is a subspace of \( \reals{m} \) with basis \( \left\{ \vblank{n}_1, \vblank{n}_2, \dots, \vblank{n}_k \right\} \), any vector in \( \nullspace{\mblank{A}} \) can be expressed as a linear combination of these basis vectors with scalars \( \alpha_1, \dots, \alpha_k \) such that
    \begin{align}
        \mblank{B} \vblank{z} = \sum_{i=1}^k \alpha_i \vblank{n}_i.
    \end{align}
    
    By assumption, for each \( i \), the system \mbox{\( \mblank{B} \vblank{z} = \vblank{n}_i \)} has no solution. Since \( \rangespace{\mblank{B}} \) is a subspace, if it does not contain the basis vectors \( \vblank{n}_i \), it cannot contain any of their linear combinations. Therefore, the only possible solution to \mbox{$\mblank{B} \vblank{z} \in \nullspace{\mblank{A}}$} is $\mblank{B} \vblank{z} = \mathbf{0}$, implying $\vblank{z} \in \nullspace{\mblank{B}}$. However, since \( \mblank{B} \) is full rank, its nullspace is trivial, meaning $\nullspace{\mblank{A} \mblank{B}} = \left\{ \mathbf{0} \right\}$ and $\rank{\mblank{A} \mblank{B}} = n$.
\end{proof}
    
The observability conditions for the nonlinear system with three unique unit vector measurement sources is presented in the following theorem.
\begin{theorem} \label{thm:sys3}
    The observability codistribution, $\obsv{}$, is full rank and the nonlinear system \eqref{eq:statedef} is locally observable with three relative unit vector measurements if the non-collinear condition is satisfied: \mbox{$\qpos{\marker2}{\marker1}{\target} \cross \qpos{\marker3}{\marker1}{\target} \neq \qzero$}.
\end{theorem}
    
\begin{proof}
    Once more, $\obsv{}$ \eqref{eq:dounit} is a block triangular matrix with its lower diagonal block right multiplied by $\qright{\dualq{\target}{\camera}}$. The observability analysis will follow a similar approach to Section \ref{sec:position_analysis} by first proving rank of the repeated diagonal block matrix $\omat$ to then prove the overall rank of $\obsv{}$.
    
    {A basis for the nullspace of the block diagonal matrix \( \pmat \) } can readily be constructed from the nullspaces of each diagonal block:
    \begin{align}
        \nullspace{\pmat} &= {\text{span} } \left\{
        \mtx{\qblank{p} \\ \qzero \\ \qzero \\ \qzero},
        \mtx{\qzero \\ \qpos{\marker1}{\camera}{\camera} \\ \qzero \\ \qzero},
        \mtx{\qzero \\ \qzero \\ \qpos{\marker2}{\camera}{\camera} \\ \qzero},
        \mtx{\qzero \\ \qzero \\ \qzero \\ \qpos{\marker3}{\camera}{\camera}}
        \right\} \\
        &= {\text{span}}\left\{ \vblank{n}_1, \vblank{n}_2, \vblank{n}_3, \vblank{n}_4 \right\},
    \end{align}
    where \mbox{$\qblank{p} \in \nullspace{\qfunc{\q{\target}{\camera}}}$} is represented as a quaternion and equivalently satisfies \mbox{$\q{\target}{\camera} \qdot \qblank{p} = \qzero$}. Additionally, one can directly verify \mbox{$\nullspace{\jacobian \ph{\qpos{\marker}{\camera}{\camera}}} = \left\{ \qpos{\marker}{\camera}{\camera} \right\}$} is a direct consequence of the position vector being projected/normalized along the direction of the vector. Now investigate the cases for each basis vector \mbox{\( \vblank{n}_i \in \nullspace{\pmat} \)} such that the system $\lmat \vblank{z} = \vblank{n}_i$, where $\vblank{z}_1, \vblank{z}_2 \in \quats$ and $\vblank{z} = \mtx{\vblank{z}_1 & \vblank{z}_2} \in \quats \times \quats$, has no solution.
    
    \textbf{Case 1:} For \( \vblank{n}_1 = \mtx{\qblank{p} & \qzero & \qzero & \qzero} \), the system of equations $\lmat \vblank{z} = \vblank{n}_1$ is:
    \begin{gather}
        \mtx{\eye{4} & \zeros{4}{4} \\
        \kfunc{\q{\target}{\camera} \qpos{\marker1}{\target}{\target}} + \qleft{\qmu{\target}{\camera}} \eyeconj & \qright{\qconj{\q{\target}{\camera}}} \\
        \kfunc{\q{\target}{\camera} \qpos{\marker2}{\target}{\target}} + \qleft{\qmu{\target}{\camera}} \eyeconj & \qright{\qconj{\q{\target}{\camera}}} \\
        \kfunc{\q{\target}{\camera} \qpos{\marker3}{\target}{\target}} + \qleft{\qmu{\target}{\camera}} \eyeconj & \qright{\qconj{\q{\target}{\camera}}}
        }
        \vblank{z}
        =
        \mtx{\qblank{p} \\ \qzero \\ \qzero \\ \qzero}.
    \end{gather}
    The bottom three block rows yield $\dmat \vblank{z} = \mtx{\qzero & \qzero & \qzero}$. By Theorem~\ref{thm:obsv_pos}, $\dmat$ is full rank when the markers are non-collinear. Applying this same condition, which gives the non-collinear observability condition of Theorem~\ref{thm:sys3}, it follows that \mbox{$\dmat \vblank{z} = \mtx{\qzero & \qzero & \qzero}$} implies $\vblank{z} = \mtx{\qzero & \qzero}$. Substituting into the top block row equation gives \( \eye{4} \qzero = \qblank{p} \), which contradicts the requirement that \( \qblank{p} \neq \qzero \) for $\vblank{n}_1$ to be a non-trivial nullspace basis vector of $\nullspace{\pmat}$. Therefore, \mbox{\( \vblank{n}_1 \notin \rangespace{\lmat} \)} when the markers are non-collinear.

    \textbf{Cases 2-4:} For \( \vblank{n}_2, \vblank{n}_3, \vblank{n}_4 \), each basis vector is the respective position vector of the marker relative to the camera represented in the camera frame of reference. Looking just at $\vblank{n}_2$ for now, the first block row equation of $\lmat \vblank{z} = \vblank{n}_2$ reduces to $\eye{4} \vblank{z}_1 = \qzero$, implying $\vblank{z}_1 = \qzero$. The bottom three block row equations then reduce to:
    \begin{gather}
        \mtx{\qright{\qconj{\q{\target}{\camera}}} \\
        \qright{\qconj{\q{\target}{\camera}}} \\
        \qright{\qconj{\q{\target}{\camera}}}
        }
        \vblank{z}_2
        =
        \mtx{\qpos{\marker1}{\camera}{\camera} \\ \qzero \\ \qzero}.
    \end{gather}
    Since $\qright{\qconj{\q{\target}{\camera}}}$ is full rank \cite{Andrews2024-jk}, the second and third equations then require \( \vblank{z}_2 = \qzero \). Substituting this into the first equation yields \( \qzero = \qpos{\marker1}{\camera}{\camera}\), which violates the positive depth constraint of the unit vector measurement model. The same approach can be repeated for \( \vblank{n}_3 \) and \( \vblank{n}_4 \) to conclude that \mbox{\( \vblank{n}_2, \vblank{n}_3, \vblank{n}_4 \notin \rangespace{\lmat} \)}.

    In summary, if the non-collinear condition is satisfied, then $\lmat$ is full rank with $\vblank{n}_1, \vblank{n}_2, \vblank{n}_3, \vblank{n}_4 \notin \rangespace{\lmat}$. It follows from Lemma~\ref{lem:rank} that $\omat$ is full rank, and consequently $\obsv{}$ is full rank, establishing local observability of the system.
\end{proof}

The observability condition derived in Theorem~\ref{thm:sys3} recovers the well-known collinear-feature degeneracy from the \ptp literature, providing a direct connection between algorithmic and nonlinear observability perspectives on the problem. Classical \ptp solvers are known to admit up to four geometrically valid solutions even under non-degenerate (non-collinear) configurations, conventionally disambiguated by a fourth point correspondence. Theorem~\ref{thm:sys3} provides a complementary view: the four \ptp solutions represent a global multiplicity, while local weak observability holds at each solution individually. Consequently, once an initial pose estimate places the filter near the true solution, three measurements suffice to uniquely determine subsequent poses, with the prior estimate playing the disambiguating role traditionally assigned to a fourth point. 

While the deterministic nonlinear system is observable when the observability conditions are satisfied, in practice, nearly unobservable conditions, such as nearly collinear markers or large process or measurement noises, can render the system effectively unobservable. These practical challenges motivate the design of observability-optimal marker placement \cite{andrews2025optimalfiducialmarkerplacement}, trajectory planning, and estimators to minimize these effects.

\section{Dual Quaternion Estimator Design} \label{sec:estimator}
In this section, a dual quaternion unscented Kalman filter (DQ-UKF) {is presented} for estimating the relative pose and velocity of a rigid body. This architecture is a specific instantiation of the Lie algebraic UKF \cite{Brossard2017-nf} in which the state on \se is represented as a dual quaternion, and extends \cite{Li2021-pg} by incorporating dual velocity state variables. This design is motivated by two main reasons. First, although the UKF introduces a higher computational cost compared to an EKF, it typically provides better performance on nonlinear systems and avoids the need to derive analytic Jacobians, which can become intractable for complex systems. Second, an invariant filtering formulation {is adopted here} that respects the manifold constraints of the dual quaternion pose representation. These design choices are made to maximize tracking performance at the cost of processing time. Despite the computational overhead, we expect real-time feasibility in practice, and the formulation is presented in a general manner so that practitioners may simplify components based on application needs. The remainder of this section is a high-level presentation of the DQ-UKF applied to the relative motion dynamics \eqref{eq:eom} with the accompanying pseudocode and code repository provided in the Appendix. 

A ``loosely coupled'' estimation framework {is adopted} in which the target’s relative motion is estimated independently of the camera’s inertial state. This design choice is motivated by the fact that autonomous vehicles almost always maintain their own inertial estimates using an onboard or offboard state estimator, and {we} do not wish to make any assumptions about its implementation. Additionally, prior work \cite{Eckenhoff2020-wz, Eckenhoff2019-vj} has shown that while ``tightly coupled'' estimators—where target and primary vehicle states are estimated jointly in the state vector—can outperform ``loosely coupled'' alternatives, this benefit only occurs when the target process model accurately reflects the true dynamics. If the target's motion model is inaccurate, the tightly coupled approach can corrupt the primary vehicle’s inertial estimate as target dynamics errors propagate through the joint filter. Since the target may be uncooperative and its inertial properties may be only approximately known, a loosely coupled structure is more robust in practice. This formulation also enables running one estimator instance per tracked target in parallel, which can be instantiated or deleted as needed.



The state estimate, $\statest$, is the same definition as in \eqref{eq:statedef} and is sometimes referred to as the Lie group or manifold state estimate in the literature. The error state and covariance, \mbox{$\lstate \in \reals{12}$} and \mbox{$\cov \in \mathbb{S}_{++}^{12}$}, are defined as
\begin{gather}
    \lstate =
    \begin{bmatrix}
    \overline{\phi}_\Delta &
    \overline{r}_\Delta &
    \overline{\omega}_\Delta &
    \overline{v}_\Delta
    \end{bmatrix}, \quad
    \lstate \sim \normal(\zeros{12}{},\cov),
\end{gather}
where $\overline{\phi}_\Delta, \overline{r}_\Delta, \overline{\omega}_\Delta, \overline{v}_\Delta \in \mathbb{R}^3$ are the rotation, position, angular velocity, and translational velocity error vectors, respectively. The error states are represented in the $\camera$ frame for this particular application. Following the dual pose error definition of \cite{Filipe2015-ep}, define the dual pose error, $\dualqblank{q}_\Delta \in \dquatsu$, of the estimate relative to the truth as
\begin{gather}
    \dualqblank{q}_\Delta {\triangleq} \qconj{\dualqblank{q}_\text{truth}} \dualqblank{q}_\text{estimate}, 
\end{gather}
which is a left-invariant Lie group error definition because
\begin{gather}
    \qconj{(\dualqblank{a} \dualqblank{q}_\text{truth})}(\dualqblank{a} \dualqblank{q}_\text{estimate})
= \qconj{\dualqblank{q}_\text{truth}} \dualqblank{q}_\text{estimate}
\end{gather}
for $\dualqblank{a} \in \dquatsu$. This dual quaternion error definition is used both when mapping UKF sigma points from the Lie algebra to the dual quaternion group and when applying the measurement update to the dual quaternion state estimate.

If $\lstate$ is a sigma point in the error space, the corresponding dual quaternion sigma point state, $\sigmap \in \dquatsu \times \dquatsv$, is
\begin{gather}
    \sigmap = \textsc{retract}(\statest, \lstate).
\end{gather}
As described in the pseudocode, dual pose perturbations are applied multiplicatively, whereas dual velocity perturbations are applied additively. However, the dual velocity cannot be perturbed as a single quantity because the dual part contains a bilinear coupling term $\qblank{\omega} \times \qblank{r}$ between the angular velocity and position, arising from expressing the velocity in the rotating camera frame. Perturbing the full dual velocity directly would produce an inconsistent error state, as it would not correctly account for how perturbations in $\qblank{r}$ and $\qblank{\omega}$ propagate through the transport term. Instead, the angular and linear velocity components are perturbed individually, and the dual velocity is reassembled from the updated components.

Each dual quaternion sigma point is then propagated through the dynamics \eqref{eq:eom}. A closed-form solution to the process model exists if $\ddualomega{\target}{\camera}{\camera} = \dualzero$, but in general this term may be nonzero and require numerical integration. Standard numerical integration methods typically violate the unit-norm constraint of quaternions and dual quaternions, so the propagated state is explicitly renormalized following the procedure outlined in $\textsc{normalize}(\cdot)$ \cite{Filipe2015-ep, Zivan2022-pw}. Although numerical integration errors can accumulate over long time horizons, for real-time update rates on the order of 100 Hz these errors become negligible and standard fourth-order Runge–Kutta integration performs comparably to more advanced Lie algebraic integrators at such step sizes \cite{Andrle2012-sg}. Thorough analyses on the effects of quaternion normalization on estimate error can be found in \cite{Zivan2022-pw, Bar-Itzhack1985-fh, Bar-Itzhack1991-vo}.

In a standard UKF implementation, the mean of the propagated sigma points would be computed by a weighted Euclidean mean. On a manifold, however, such an average does not preserve the group structure, and computing a proper group mean requires solving a nonlinear optimization problem. To avoid this additional computational cost, the propagated mean {is approximated} using the zero-perturbation sigma point. The state differences between each propagated sigma point and the mean are then lifted back to the Lie algebra via $\lstate = \textsc{lift}(\statest, \sigmap)$, and the covariance is updated accordingly.

Upon receiving a measurement, $\meas \in \reals{k}$, new sigma points are generated in the error space, retracted to the manifold, and passed through the measurement function. If unit quaternion or unit dual quaternion measurements were used instead of Euclidean vectors, care would be required to compute innovation differences on the manifold instead. In our implementation, the Kalman gain is computed in the standard manner, the measurement update $\lstate_\star \in \reals{12}$ is calculated in the Lie algebra, and the dual quaternion state estimate is updated via
\begin{gather}
    \statest_k = \textsc{retract}(\statest_{k|k-1}, \lstate_\star).
\end{gather}


\section{Numerical Example} \label{sec:simulation}
The analytical observability conditions and utility of the DQ-UKF will now be demonstrated {via a numerical example}. Specifically, consider a visual target tracking scenario in which both the camera and target undergo acceleration relative to each other, and the camera estimates the target's relative pose and velocity states using only a pinhole camera measurement model. A loosely coupled framework {was assumed} and that the camera's pose and velocities with respect to the inertial frame were known. Pinhole camera measurements and \pnp pose estimates were obtained using OpenCV~\cite{2015opencv}, with the \textsc{SolvePnP} function employing the algorithm of~\cite{terzakis2020}. OpenCV pose estimates were used as a baseline for comparison against DQ-UKF results.

Previous works have thoroughly analyzed the sensitivity of memoryless deterministic \pnp solvers to measurement noise and feature point count~\cite{Zheng2013-of, Ferraz2014-yb}, and \cite{Mehralian2020-ns} extended this analysis to an EKF-based \pnp pose estimator, comparing its performance against traditional solvers. Hardware demonstrations have further validated the utility of \pnp filters in applied settings~\cite{Sveier2021-ev, Tweddle2015-jg, Filipe2015-ep}. The present results complement this body of work by quantitatively demonstrating the robustness of the DQ-UKF, and by extension a broader class of sequential visual pose estimators, to occlusions, showing that sequential estimators maintain lower estimation error compared to memoryless \pnp solvers.

The process model used in the filter was the relative dynamics from \eqref{eq:eom}. A zero-order hold {was assumed} on $\ddualomega{\target}{\inertial}{\target}$ and $\ddualomega{\camera}{\inertial}{\camera}$ over the state propagation time interval. This design choice was motivated by real-time implementation considerations: typically, these quantities vary slowly over the sample period, and repeated evaluation of the acceleration models during numerical integration introduces undesirable computational complexity. Errors introduced by the zero-order hold and post-propagation dual quaternion state normalization were compensated through the addition of a small amount of process noise applied to the state estimate covariance.

An overview of the 3-D trajectories of the camera and target frames with respect to the inertial frame are shown in Fig. \ref{fig:frames}. The trajectories are representative of a flyby or interception between the camera and target. A time history of the true relative pose and velocity states in the camera frame is shown in Fig. \ref{fig:truth}. The target was spinning and translating about its -Z axis at constant velocities, and the camera was experiencing a small force in the body frame +X and torque along +Y, resulting in a translation and rotation towards the approaching target.
\begin{figure}[t]
    \centering
    \includegraphics[width=\linewidth]{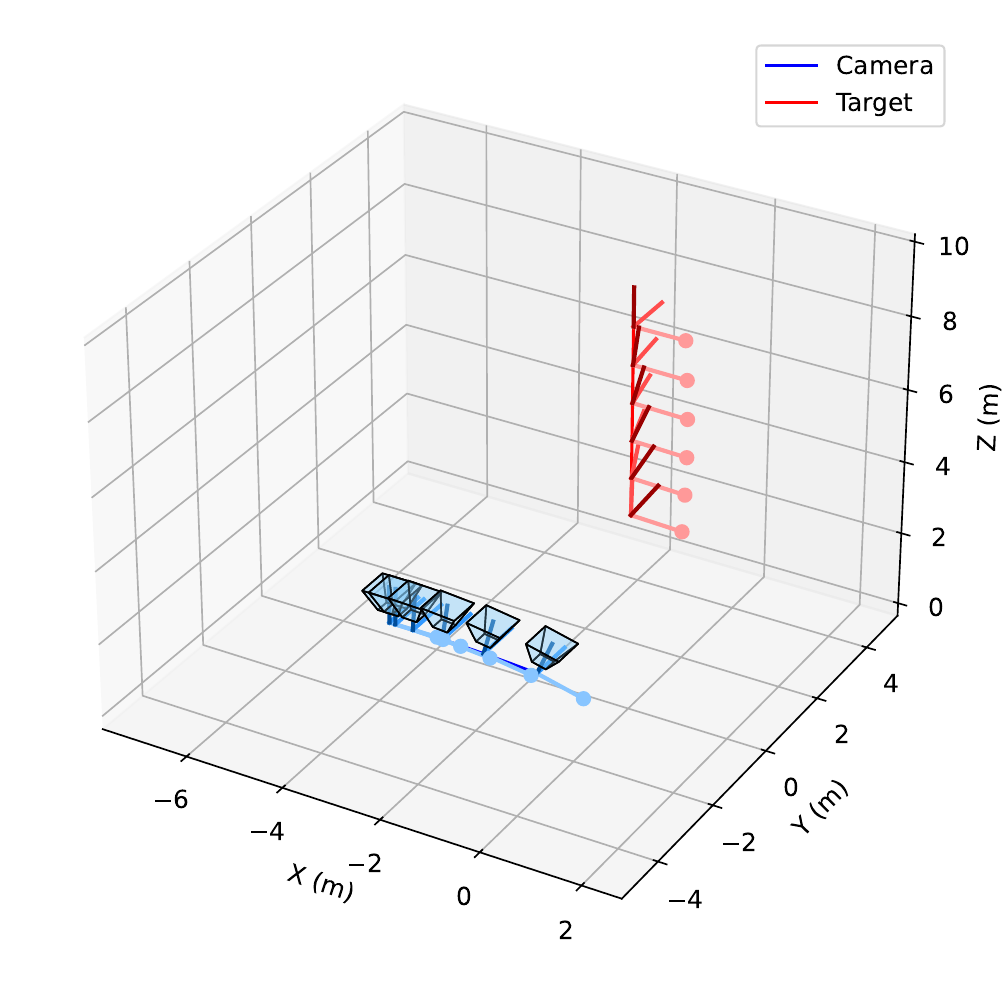}
    \caption{3-D plot showing the camera and target coordinate frame trajectories with respect to an inertial frame evolving and approaching each other over time. The +X-axis of each body fixed frame is denoted with a sphere on the end and the camera is affixed to the camera's \mbox{+Z-axis} and denoted with a frustum.}
    \label{fig:frames}
\end{figure}
\begin{figure}[t]
    \centering
    \includegraphics[width=\linewidth]{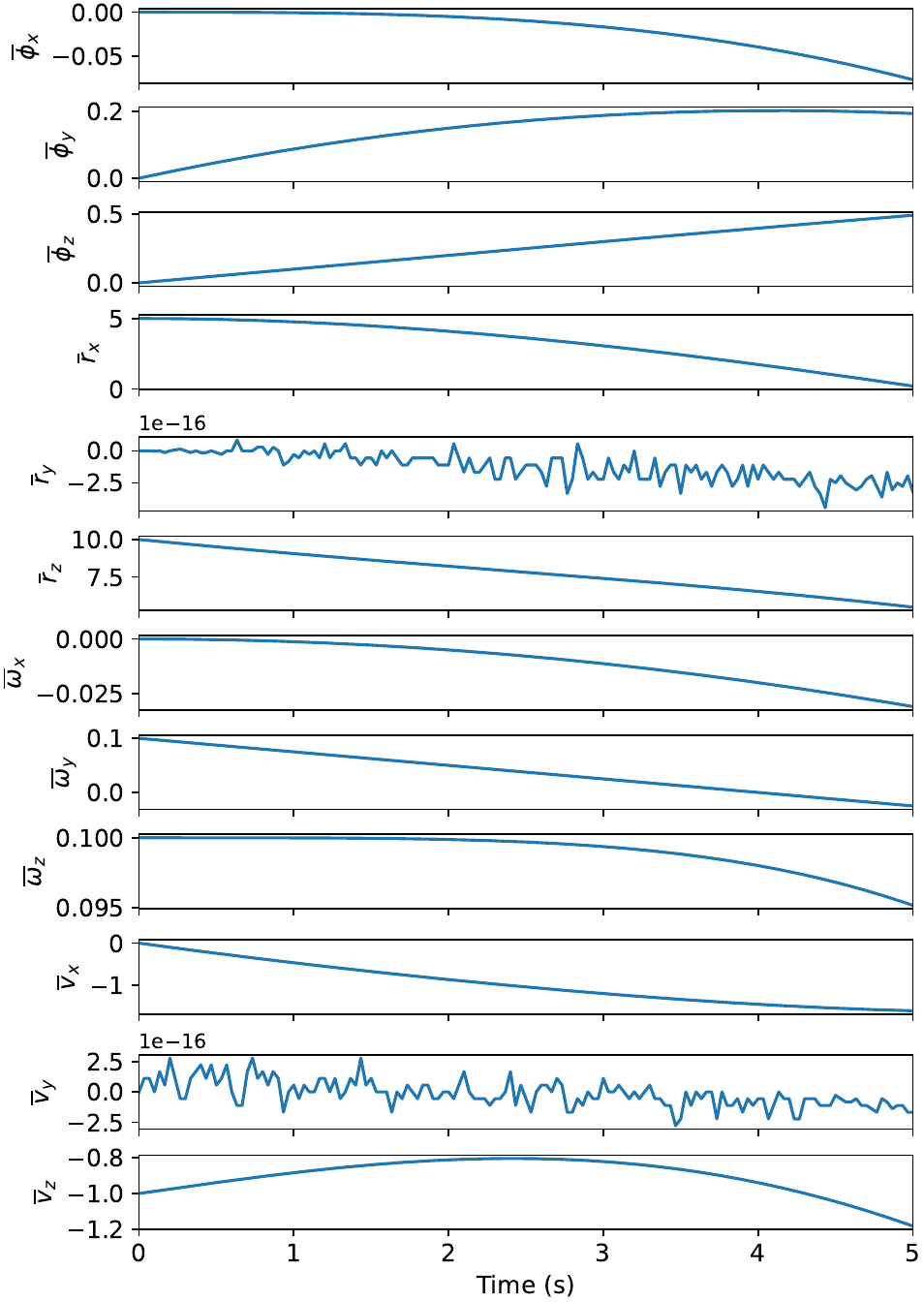}
    \caption{Time history of the truth pose and velocities of the target relative to camera in camera frame coordinates. All states are in base SI units.}
    \label{fig:truth}
\end{figure}

Five marker points were placed on the target's X-Y plane, one at the origin and the other four at the four corners of the box drawn by $(\pm1, \pm1)$. Using OpenCV, the feature points were projected through the pinhole camera model with the following camera intrinsic matrix
\begin{gather}
    \mathbf{K} = \mtx{800 & 0 & 640 \\ 0 & 800 & 512 \\ 0 & 0 & 1}.
\end{gather}
The measured features were then corrupted with zero-mean camera pixel sensor noise defined by the standard deviation $\sigma_{\meas}$. Since each projected feature returns a 2-D pixel measurement, the measurement covariance matrix for $m$ feature points is $\mathbf{R} = \sigma_{\meas}^2 \eye{2m}$. A small amount of process noise was applied to the velocity states to compensate for the zero-order hold of the camera and target's accelerations and to better replicate the qualitative estimator performance one might see in practice. The process noise covariance was the block diagonal matrix \mbox{$\mathbf{Q} = \Delta t \cdot \blkdiag{\sigma_{\overline{\phi}}^2 \eye{3}, \ \sigma_{\overline{r}}^2 \eye{3}, \ \sigma_{\overline{\omega}}^2 \eye{3}, \ \sigma_{\overline{v}}^2 \eye{3}}$}, where $\Delta t$ was the time step duration between filter updates. Process and measurement noise values are shown in Table \ref{table:cov_q}.

An initial pose estimate was obtained from measurements using OpenCV's \pnp solver, while the initial velocity estimate was randomly sampled from a normal distribution centered at the true velocity with covariance given by the initial estimate covariance in Table \ref{table:cov_p}. While a more sophisticated and realistic initialization method could be employed to extract velocity estimates from a sequence of \pnp solutions, it was not the focus of this work and we {opted} for a simpler, yet still representative, approach.
\begin{table}[t]
    \caption{Process and measurement noise standard deviations.}
    \label{table:cov_q}
    \centering
    \begin{tabular}{clc}
        \toprule
        \textbf{Variable} & \textbf{Description} & \textbf{Value} \\
        \midrule
        $\sigma_{\overline{\phi}}$ & Rotation vector (rad)            & $0$   \\
        $\sigma_{\overline{r}}$   & Position (m)                     & $0$   \\
        $\sigma_{\overline{\omega}}$ & Angular velocity (rad/s)      & $0.3$ \\
        $\sigma_{\overline{v}}$   & Velocity (m/s)                   & $0.3$ \\
        $\sigma_{\meas}$          & Pinhole camera measurement (px)  & $2$   \\
        \bottomrule
    \end{tabular}
\end{table}
\begin{table}[t]
    \caption{Initial state estimate standard deviations.}
    \label{table:cov_p}
    \centering
    \begin{tabular}{lc}
        \toprule
        \textbf{Description} & \textbf{Value} \\
        \midrule
        Rotation vector (rad)    & $0.2$ \\
        Position (m)             & $0.2$ \\
        Angular velocity (rad/s) & $0.1$ \\
        Velocity (m/s)           & $0.5$ \\
        \bottomrule
    \end{tabular}
\end{table}

The scenario was simulated over five seconds with measurements updated at 30 Hz. To investigate pose estimation performance under occlusions, the number of active markers {varied} over a 5-second sequence: four markers (0-1s), three markers (1-2s), two markers (2-3s), three markers (3-4s), and four markers (4-5s). DQ-UKF estimate errors with $\pm 3\sigma$ covariance bounds are shown in Fig. \ref{fig:ukf}. Figure \ref{fig:compare} compares the orientation and position estimate error magnitudes against OpenCV's \pnp solver as a baseline. A pose estimate was unavailable between 2-3s because the OpenCV \pnp solver requires a minimum of three points.
\begin{figure}[t]
    \centering
    \includegraphics[width=\linewidth]{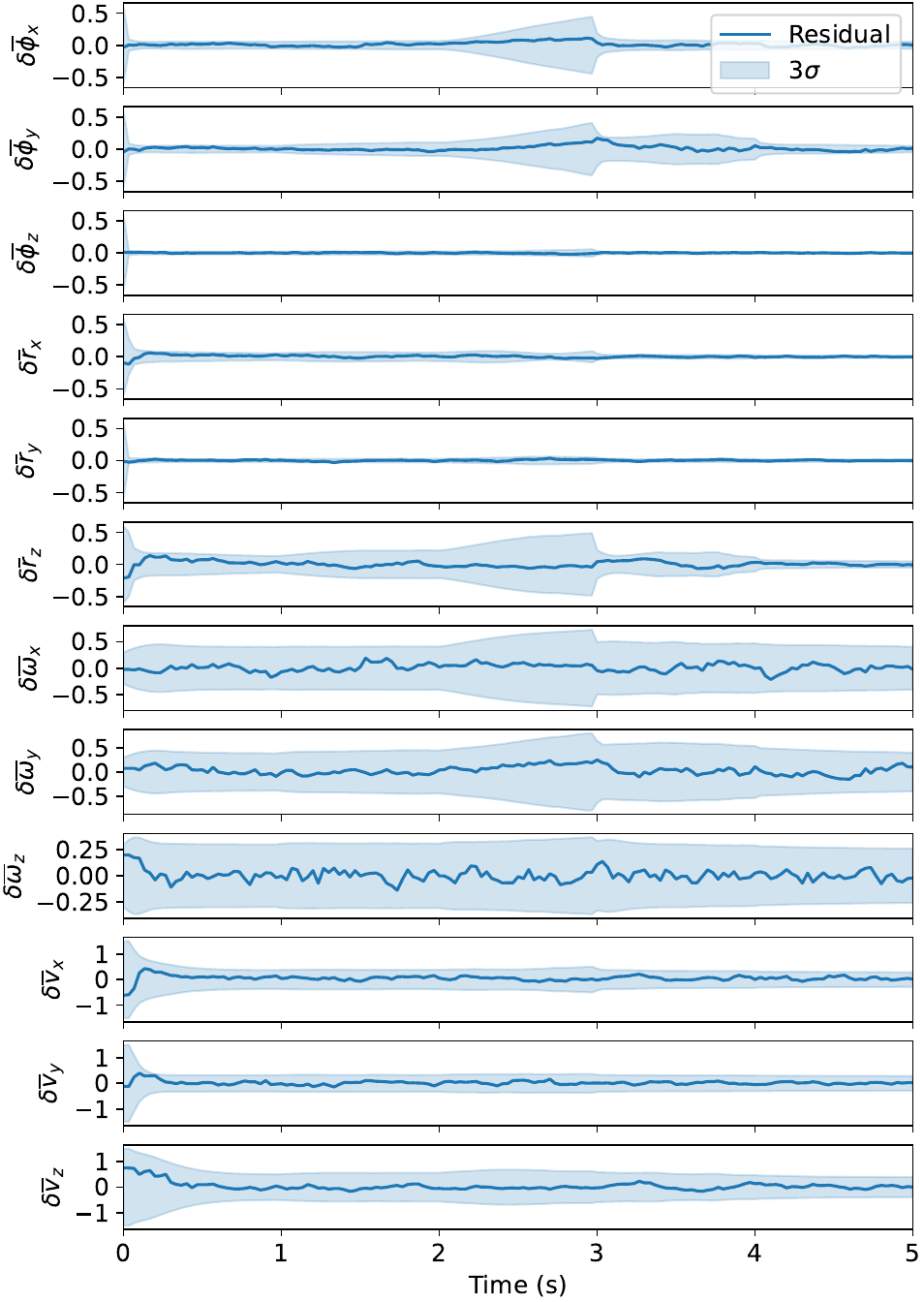}
    \caption{Time history of DQ-UKF estimate errors and estimate $3\sigma$ covariance bounds. All states are in base SI units.}
    \label{fig:ukf}
\end{figure}
\begin{figure}[t]
    \centering
    \includegraphics[width=\linewidth]{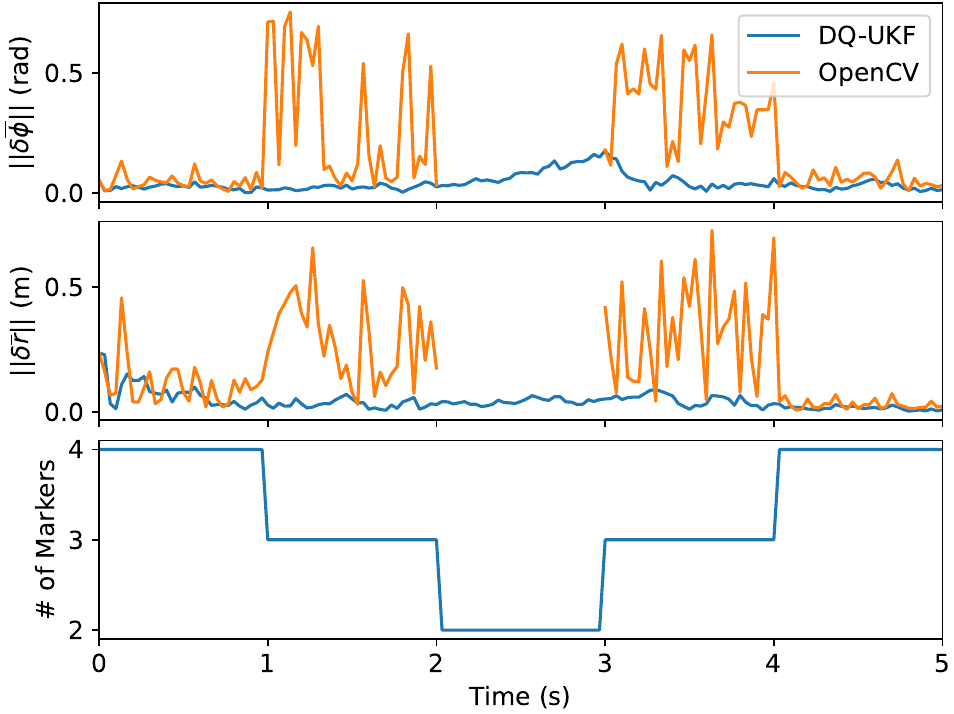}
    \caption{Pose estimation error comparison between DQ-UKF and OpenCV \pnp solver: orientation error (top), position error (middle), and number of measured markers (bottom). Occlusions are simulated by decreasing and increasing the number of markers used for pose estimation.}
    \label{fig:compare}
\end{figure}

Several key insights emerge from these simulation results. During the 2-3s interval when only two markers were visible, the system became unobservable according to Theorem \ref{thm:sys3}. This unobservability was evident in Fig. \ref{fig:ukf}, where the estimate covariance {increased} for some states. Interestingly, the covariance did not increase uniformly across all states, indicating that two markers still {provided} partial information about a subset of state variables depending on the geometry between the target and camera. This {result} is a subtle but important distinction: while the observability conditions from Theorem \ref{thm:sys3} apply to the entire state, fewer measurements can still render a subset of state variables observable, as illustrated by these results. Additionally, Fig. \ref{fig:compare} shows that the DQ-UKF {achieved} lower and smoother pose estimation errors on average compared to a memoryless \pnp solver across all simulated marker counts.

Both advantageous characteristics, smaller estimate error and graceful handling of unobservable periods, stem from incorporating a motion model and prior estimates. Unlike a memoryless \pnp solver, which cannot provide solutions with fewer than three markers, sequential state estimators, such as the DQ-UKF, maintain and propagate estimates under significant occlusions. These results demonstrate that \pnp estimators can achieve lower estimation error and greater robustness to occlusions when additional modeling and computational overhead can be incorporated.


\section{Conclusion} \label{sec:conclusion}
This article presented a comprehensive modeling, analysis, and filter design for dynamic 6-DOF target tracking using a dual quaternion representation. Specifically, the work presented here was motivated by visual pose and velocity tracking of a target. In Sections \ref{sec:position_analysis} and \ref{sec:proj_analysis}, sufficient conditions for local observability of the nonlinear relative motion dynamic system were derived using first relative position measurements and then relative unit vector measurements. Notably, the observability conditions found in Section \ref{sec:proj_analysis} are identical to the \ptp minimal conditions for a solution found in memoryless \ptp solvers, offering a novel control-theoretic connection between local observability and multiple \ptp solutions. These analyses further demonstrated the utility of a dual quaternion modeling framework for 6-DOF system analysis.

Following the observability analyses, a Lie algebraic UKF for pose and velocity estimation was formalized in a dual quaternion formulation and presented with accompanying pseudocode and code repository. Lastly, simulation results for a visual target tracking scenario with non-zero relative accelerations demonstrated that the DQ-UKF achieves superior tracking performance compared to a memoryless \pnp solver, particularly in the presence of occlusions and measurement noise, while having the additional benefit of maintaining estimate covariance for uncertainty quantification and decision making purposes.

This work identified several promising directions for future research. A fundamental open question concerns the optimal placement of visual markers and selection of feature points for pose estimation. While increasing the number of visual feature points generally improves pose estimates, the computational burden of detecting and correlating large feature sets grows substantially, with asymptotically diminishing returns in estimation accuracy. By incorporating physics-based models of relative motion, the proposed filter architecture may achieve comparable tracking performance with substantially reduced feature point requirements. Observability-based optimal sensor placement methods~\cite{andrews2025optimalfiducialmarkerplacement, 
Hinson2015-se, brace22, Smocot2026-da} have shown promise in addressing analogous problems in other domains and represent a natural control-theoretic direction for minimizing feature point requirements in visual pose estimation.
More broadly, we anticipate this work will facilitate broader adoption of dual quaternions for 6-DOF system modeling and analysis, catalyzing development of novel estimation, control, and planning algorithms that leverage their compact, singularity-free geometric properties.

\appendix
\label{sec:appendix}
The pseudocode below adheres to NumPy's zero-based indexing and bracket notation for arrays. Sigma point generation is based on the method in \cite{Wan2002-sr}. Python implementations of the pseudocode and Section \ref{sec:simulation} simulation code are available at \texttt{\detokenize{github.com/uwaa-ndcl/dqukf_pose_estimation}}.
\begin{algorithmic}[1]
\Function{normalize}{$\statest$}
    \State $\dualqblank{q}, \ \dualqblank{\omega} = \statest[0], \ \statest[1] \in \dquats$
    \State $\qreal{q}, \qdual{q} = \qreal{(\dualqblank{q})}, \qdual{(\dualqblank{q})} \in \quats$
    \State $\qreal{\omega}, \qdual{\omega} = \qreal{(\dualqblank{\omega})}, \qdual{(\dualqblank{\omega})} \in \quats$
    \State $\qreal{q} = \qreal{q} / \norm{\qreal{q}} \in \quatsu$
    \State $\qdual{q} = \qdual{q} - \qreal{q} \left(\qdual{q} \qdot \qreal{q} \right) \in \quats$
    \State $\qreal{\omega} = \mtx{0 & \overline{\omega}_r} \in \quatsv$
    \State $\qdual{\omega} = \mtx{0 & \overline{\omega}_d} \in \quatsv$
    \State $\dualqblank{q} = \qreal{q} + \dualunit \qdual{q} \in \dquatsu$
    \State $\dualqblank{\omega} = \qreal{\omega} + \dualunit \qdual{\omega} \in \dquatsv$
    \State \Return $\mtx{\dualqblank{q} & \dualqblank{\omega}} \in \dquatsu \times \dquatsv$
\EndFunction
\Statex

\Function{retract}{$\statest, \lstate$} \label{func:retract}
    \State $\dualqblank{q}, \ \dualqblank{\omega} = \statest[0], \ \statest[1] \in \dquatsu, \dquatsv$
    \State ${\overline{\phi}_\Delta, \overline{\qblank{r}}_\Delta, \overline{\qblank{\omega}}_\Delta, \overline{\qblank{v}}_\Delta} = \lstate[:3], \lstate[3:6], \lstate[6:9], \lstate[9:] \in \reals{3}$
    \State $\theta = \norm{\overline{\phi}_\Delta}_2 \in \reals{}$
    \State $\qblank{q}_\Delta = \mtx{\cos \left( \frac{\theta}{2} \right) & \frac{\overline{\phi}_\Delta}{\theta} \sin \left( \frac{\theta}{2} \right)} \in \quatsu$
    \State $\qblank{r}_\Delta = \mtx{0 & \overline{r}_\Delta} \in \quatsv$
    \State $\qblank{\omega}_\Delta = \mtx{0 & \overline{\omega}_\Delta} \in \quatsv$
    \State $\qblank{v}_\Delta = \mtx{0 & \overline{v}_\Delta} \in \quatsv$
    \State $\qblank{r} = 2 \qdual{(\dualqblank{q})} \qreal{(\dualqblank{q})} \in \quatsv$
    \State $\qblank{\omega} = \qreal{(\dualqblank{\omega})} \in \quatsv$
    \State $\qblank{v} = \qdual{(\dualqblank{\omega})} + \omega \cross r \in \quatsv$
    \State $\dualqblank{q}_+ = \dualqblank{q} \left( \qblank{q}_\Delta + \dualunit \tfrac{1}{2} \qblank{r}_\Delta \qblank{q}_\Delta \right) \in \dquatsu$
    \State $\dualqblank{\omega}_+ = (\qblank{\omega} + \qblank{\omega}_\Delta) + \dualunit (\qblank{v} + \qblank{v}_\Delta)$
    \Statex \hspace{\algorithmicindent} $\qquad - (\qblank{\omega} + \qblank{\omega}_\Delta) \cross (\qblank{r} + \qblank{r}_\Delta) \in \dquatsv$
    \State \Return $\mtx{\dualqblank{q}_+ & \dualqblank{\omega}_+} \in \dquatsu \times \dquatsv$
\EndFunction
\Statex

\Function{lift}{$\statest, \statest_+$}
    \State $\dualqblank{q}, \ \dualqblank{\omega} = \statest[0], \ \statest[1] \in \dquatsu, \dquatsv$
    \State $\dualqblank{q}_+, \ \dualqblank{\omega}_+ = \statest_+[0], \ \statest_+[1] \in \dquatsu, \dquatsv$
    \State $\qblank{r} = 2 \qdual{(\dualqblank{q})} \qreal{(\dualqblank{q})} \in \quatsv$
    \State $\qblank{r}_+ = 2 \qdual{(\dualqblank{q}_+)} \qreal{(\dualqblank{q}_+)} \in \quatsv$
    \State $\qblank{\omega} = \qreal{(\dualqblank{\omega})} \in \quatsv$
    \State $\qblank{\omega}_+ = \qreal{(\dualqblank{\omega}_+)} \in \quatsv$
    \State $\qblank{v} = \qdual{(\dualqblank{\omega})} + \omega \cross r \in \quatsv$
    \State $\qblank{v}_+ = \qdual{(\dualqblank{\omega}_+)} + \omega_+ \cross r_+ \in \quatsv$
    \State $\qblank{q}_\Delta = \qreal{(\qconj{\dualqblank{q}} \dualqblank{q}_+)} \in \quatsu$
    \State $\qscalar{q}, \qvec{q} = \qblank{q}_\Delta[0], \qblank{q}_\Delta[1:] \in \reals{}, \reals{3}$
    \State $\overline{\phi}_\Delta = 2 \arctan(\norm{\qvec{q}}_2, \qscalar{q}) \ \frac{\qvec{q}}{\norm{\qvec{q}}_2} \in \reals{3}$
    \State $\qblank{r}_\Delta = \qblank{r}_+ - \qblank{r} \in \quatsv$
    \State $\qblank{\omega}_\Delta = \qblank{\omega}_+ - \qblank{\omega} \in \quatsv$
    \State $\qblank{v}_\Delta = \qblank{v}_+ - \qblank{v} \in \quatsv$
    \State \Return $\mtx{\overline{\phi}_\Delta & \overline{\qblank{r}}_\Delta & \overline{\qblank{\omega}}_\Delta & \overline{\qblank{v}}_\Delta} \in \reals{12}$
\EndFunction

\Statex
\Statex
\Statex
\Statex
\Statex
\Class{DualQuaternionUKF}
    \State \textbf{Attributes:} State estimate $\statest \in \dquatsu \times \dquatsv$, estimate covariance $\cov \in \setpd{12}$, process model \mbox{$f: \dquatsu \times \dquatsv \rightarrow \dquatsu \times \dquatsv$}, measurement model $h: \dquatsu \times \dquatsv \rightarrow \reals{k}$, process noise $\mathbf{Q} \in \setpd{12}$, measurement noise $\mathbf{R} \in \setpd{k}$, sigma point parameters $\alpha, \beta, \kappa \in \reals{}$
    \Statex
    
    \Procedure{\_\_init\_\_}{$\statest_0, \cov_0, f, h, \mathbf{Q}, \mathbf{R}, \alpha, \beta, \kappa$}
        \State $\statest, \cov, f, h, \mathbf{Q}, \mathbf{R}, n = \statest_0, \cov_0, f, h, \mathbf{Q}, \mathbf{R}, 12$
        \State $\lambda = \alpha^2 (n + \kappa) - n$
        \State $\mathbf{w}_m[0] = \frac{\lambda}{n+\lambda}$
        \State $\mathbf{w}_c[0] = \frac{\lambda}{n+\lambda} + (1-\alpha^2+\beta)$
        \State $\mathbf{w}_m[i] = \mathbf{w}_c[i] = \frac{1}{2(n+\lambda)}, \; i=1,\dots,2n$
    \EndProcedure
    \Statex
    \Procedure{sigma\_points}{}
        \State $\mathbf{S} = \sqrt{(n+\lambda) \cov}$
        \State $\lstate[0] = \zeros{n}{}$
        \For{$i=1$ to $n$}
            \State $\lstate[i] = \mathbf{S}[i-1]$
            \State $\lstate[i+n] = - \mathbf{S}[i-1]$
        \EndFor
        \State \Return $\lstate$
    \EndProcedure
    \Statex
    \Procedure{predict}{}
        \State $\lstate = \Call{sigma\_points}{}$
        \For{$i=0$ to $2n$}
            \State $\sigmap[i] = \textsc{retract}(\statest, \lstate[i])$
            \State $\sigmap[i] = f(\sigmap[i])$
            \State $\sigmap[i] = \textsc{normalize}(\sigmap[i])$
            \State $\lstate[i] = \textsc{lift}(\sigmap[0], \sigmap[i])$
        \EndFor
        \State $\statest = \sigmap[0]$
        \State $\cov = \sum_{i=0}^{2n} \mathbf{w}_c[i] \lstate[i] \lstate[i]^\top + \mathbf{Q}$
    \EndProcedure
    \Statex
    \Procedure{update}{$\meas_k$}
        \State $\lstate = \Call{sigma\_points}{}$
        \For{$i=0$ to $2n$}
            \State $\sigmap[i] = \textsc{retract}(\statest, \lstate[i])$
            \State $\bm{\gamma}[i] = h(\sigmap[i])$
        \EndFor
        \State $\hat{\meas}_k = \sum_{i=0}^{2n} \mathbf{w}_m[i] \bm{\gamma}[i]$
        \State $\mathbf{S}_k = \sum_{i=0}^{2n} \mathbf{w}_c[i](\bm{\gamma}[i] - \hat{\meas}_k)(\bm{\gamma}[i] - \hat{\meas}_k)^\top + \mathbf{R}$
        \State $\mathbf{P}_{xy} = \sum_{i=0}^{2n} \mathbf{w}_c[i] \lstate[i] (\bm{\gamma}[i] - \hat{\meas}_k)^\top$
        \State $\mathbf{K}_k = \mathbf{P}_{xy} \mathbf{S}_k^{-1}$
        \State $\lstate_\star = \mathbf{K}_k(\meas_k - \hat{\meas}_k)$
        \State $\statest = \textsc{retract}(\statest, \lstate_\star)$
        \State $\cov = \cov - \mathbf{K}_k \mathbf{S}_k \mathbf{K}_k^\top$
    \EndProcedure
\EndClass
\end{algorithmic}

\newpage
\bibliographystyle{IEEEtran}
\bibliography{IEEEabrv, references}

\end{document}